\documentclass[a4paper]{article}
\usepackage{amsfonts}
\usepackage{amssymb}
\usepackage{amsmath}
\usepackage{geometry}
\usepackage[dvipdfmx]{graphicx}

\newtheorem{theorem}{Theorem}

\newtheorem{lemma}[theorem]{Lemma}

\newtheorem{remark}[theorem]{Remark}

\newenvironment{proof}[1][Proof]{\noindent\textbf{#1.} }{\ \rule{0.5em}{0.5em}}
\input{tcilatex}
\begin{document}

\title{Space-time fractional Zener wave equation }
\author{ Teodor M. Atanackovic\thanks{%
Department of Mechanics, Faculty of Technical Sciences, University of Novi
Sad, Trg D. Obradovica 6, 21000 Novi Sad, Serbia, atanackovic@uns.ac.rs},
Marko Janev\thanks{%
Mathematical Institute, Serbian Academy of Arts and Sciences, Kneza Mihaila
36, 11000 Belgrade, Serbia, janev.marko@gmail.com}, \\
Ljubica Oparnica\thanks{%
Faculty of Education in Sombor, University of Novi Sad, Podgori\v{c}ka 4,
25000 Sombor, Serbia, ljubica.oparnica@gmail.com}, Stevan Pilipovic\thanks{%
Department of Mathematics, Faculty of Natural Sciences and Mathematics,
University of Novi Sad, Trg D. Obradovica 4, 21000 Novi Sad, Serbia,
stevan.pilipovic@dmi.uns.ac.rs}, Dusan Zorica\thanks{%
Mathematical Institute, Serbian Academy of Arts and Sciences, Kneza Mihaila
36, 11000 Belgrade, Serbia, dusan\textunderscore zorica@mi.sanu.ac.rs}}
\maketitle

\begin{abstract}
Space-time fractional Zener wave equation, describing viscoelastic materials
obeying the time-fractional Zener model and the space-fractional strain
measure, is derived and analyzed. This model includes waves with finite
speed, as well as non-propagating disturbances. The existence and the
uniqueness of the solution to the generalized Cauchy problem are proved.
Special cases are investigated and numerical examples are presented.


\end{abstract}

\section{Introduction \label{sec:intro}}

The aim of this study is a class of generalized wave equation. Wave equation
can be generalized within the theory of fractional calculus by replacing the
second order derivative (space and/or time) with the fractional ones, as
done in \cite%
{APZ-1,APZ-2,Hanyga-stf,Kochubei08,MainardiPagniniGorenflo,MaPaMuGo,Nab}.
Space-time fractional Zener wave equation represents a generalization of the
classical wave equation obtained as a system consisting of the equation of
motion of the deformable (one-dimensional) body, the time-fractional Zener
constitutive equation and the space-fractional strain measure. Our
generalization is done by the fractionalization in both space and time
variable on the ground of the physically acceptable concepts. More details
on the formulation and the mechanical background will be given in this
section, which finishes with the remark related to the analysis of our
generalization of the wave equation.

In Section \ref{sec:fund sol} we show the existence and uniqueness of
solution to the space-time fractional Zener wave equation (\ref{eq:fzwe}), (%
\ref{eq:ic}), (\ref{eq:bc}). For this purpose we use the Fourier and Laplace
transforms in the spaces of distributions, simplifying the procedure, in a
way that we have to prove the absolute convergence of certain double
integrals. The analysis presented in Section \ref{eu} concerning the
properties of solution to the space-time fractional Zener wave equation
implies that the solution kernel is%
\begin{equation*}
P\left( x,t\right) =I\left( x,t\right) -\left( \frac{\partial }{\partial t}%
J_{1}\left( x,t\right) +\frac{\partial ^{2}}{\partial t^{2}}J_{2}\left(
x,t\right) \right) \mathrm{e}^{s_{0}t},\;\;x\in
\mathbb{R}
,\;t>0,
\end{equation*}%
where $I,$ $J_{1}$ and $J_{2}$ are continuous bounded function with respect
to the space variable $x$ and continuous exponentially bounded functions
with respect to the time variable $t,$ see Theorem \ref{th:glavna}. On the
other hand, in Section \ref{sec:calc} we show, by the regularization and
quite different estimates in comparison to those used in Section \ref{eu},
that the solution to (\ref{eq:fzwe}), (\ref{eq:ic}), (\ref{eq:bc}) is given
by a distributional limit of a net of approximated solutions, which are
continuous with respect to $x\in
\mathbb{R}
,$ $t>0,$ bounded with respect to $x\in
\mathbb{R}
$ and exponentially bounded with respect to $t>0.$ Results of Section \ref%
{sec:fund sol} are justified in Section \ref{cases} by discussing the
influence of parameters $\alpha $ and $\beta $ (orders of the time and space
fractional derivatives) on the solution to (\ref{eq:fzwe}), (\ref{eq:ic}), (%
\ref{eq:bc}) and in Section \ref{sec:ex} by the numerical examples.
Mathematical background is given in Appendix \ref{sec:math prelim}.

\subsection{Model}

Recall, the classical wave equation describes the waves that occur in
elastic medium. It is obtained from the equations of the deformable body,
see \cite{a-g}. The wave equation can be written in the form of a system
which consists of three equations: equation of motion, constitutive equation
and strain measure. Unknown functions depending on time, $t>0,$ and space, $%
x\in \mathbb{R},$ variables are: displacement $u$, stress $\sigma $ and
strain $\varepsilon $. We consider an infinite viscoelastic rod
(one-dimensional body), positioned along $x$-axis, that is not under
influence of body forces. Then, the equation of motion reads
\begin{equation}
{\partial _{x}}\sigma (x,t)=\rho \,{\partial _{t}^{2}}u(x,t),\;\;x\in
\mathbb{R},\;t>0,  \label{eq:system1-d}
\end{equation}%
where $\rho >0$ denotes the (constant) density of the rod. The constitutive
equation gives the relation between stress and strain, and in the case of
elastic media it is the Hooke law. Since we consider waves occurring in
viscoelastic media, we chose the constitutive equation to be the
time-fractional Zener model%
\begin{equation}
\sigma (x,t)+\tau _{\sigma }\,{}_{0}^{C}\mathrm{D}_{t}^{\alpha }\sigma
(x,t)=E(\varepsilon (x,t)+\tau _{\varepsilon }\,{}_{0}^{C}\mathrm{D}%
_{t}^{\alpha }\varepsilon (x,t)),\;\;x\in \mathbb{R},\;t>0,
\label{eq:system2-d}
\end{equation}%
where $E$ is the generalized Young modulus (measured in $\frac{\unit{Pa}}{%
\unit{m}^{1-\beta }}$), $\tau _{\sigma }$ and $\tau _{\varepsilon }$ are
generalized relaxation times (measured in $\unit{s}^{\alpha }$) with
(thermodynamical) restriction $0<\tau _{\sigma }<\tau _{\varepsilon }$. All
three parameters are assumed to be constant. The operator ${}_{0}^{C}\mathrm{%
D}_{t}^{\alpha }$ denotes the left Caputo operator of fractional
differentiation of order $\alpha \in \left[ 0,1\right) ,$ see Appendix \ref%
{sec:math prelim}. For $\alpha =0$, the constitutive equation (\ref%
{eq:system2-d}) reduces to the Hooke law%
\begin{equation*}
\sigma =E_{r}\varepsilon ,\;\;\text{with}\;\;E_{r}=E\frac{1+\tau
_{\varepsilon }}{1+\tau _{\sigma }}.
\end{equation*}%
For $\alpha =1$, the constitutive equation (\ref{eq:system2-d}) reduces to
the classical Zener model. For more details on fractional derivatives see
\cite{Pod,SKM}. We refer to \cite{R-S-2010} for a review on the fractional
models in viscoelasticity and to \cite{AKOZ} for a systematic analysis of
the thermodynamical restrictions on parameters in such models. The strain
measure gives connection between strain and displacement. In the classical
set up, strain measure describes local deformations and it reads: $%
\varepsilon =\partial _{x}u$. Since we consider non-local effects in
material, we use the fractional model of strain measure
\begin{equation}
\varepsilon (x,t)=\mathcal{E}_{x}^{\beta }u(x,t),\;\;x\in \mathbb{R},\;t>0,
\label{eq:system3-d}
\end{equation}%
where $\mathcal{E}_{x}^{\beta }$ denotes the symmetrized Caputo fractional
derivative of order $\beta \in \left[ 0,1\right) ,$ see Appendix \ref%
{sec:math prelim}. For $\beta =1$ we obtain the classical strain measure.
Regarding the fractionalization of the strain measure, we follow the
approach presented in \cite{A-S-09}, where the symmetrized fractional
derivative is introduced in order to describe the non-local effects of the
material. Note that in \cite{AKOZ-1}, the same type of the fractional
derivative is used in the framework of the heat conduction problem of the
space-time fractional Cattaneo type equation.

One may also treat the non-locality in viscoelastic media by the different
approach. Namely, contrary to (\ref{eq:system3-d}), one may retain the
classical strain measure and introduce the non-locality in the constitutive
equation. In the classical setting it was done by Eringen, \cite{eringen}.
In the framework of the fractional calculus this approach is followed in
\cite{CCS,CZAS,DPFPSZ,DPFZ,DPZ}. The wave equation, obtained from a system
consisting of the equation of motion, fractional Eringen-type constitutive
equation and classical strain measure, is studied in \cite{CDPZ,SCC}.

The initial conditions corresponding to system (\ref{eq:system1-d}) - (\ref%
{eq:system3-d}) are%
\begin{equation}
u(x,0)=u_{0}(x),\;\;\;\;{\partial _{t}}u(x,0)=v_{0}(x),\;\;\;\;\sigma
(x,0)=0,\;\;\;\;\varepsilon (x,0)=0,\;\;x\in \mathbb{R},  \label{eq:ic-d}
\end{equation}%
where $u_{0}$ and $v_{0}$ are initial displacement and velocity, while the
boundary conditions are%
\begin{equation}
\lim_{x\rightarrow \pm \infty }u(x,t)=0,\;\;\;\;\lim_{x\rightarrow \pm
\infty }\sigma (x,t)=0,\;\;t>0.  \label{eq:bc-d}
\end{equation}%
Note that boundary conditions (\ref{eq:bc-d}) are the natural choice for the
case of the unbounded domain, while in the case of the bounded domain there
can be a large variety of different boundary conditions depending on the
type of problem one faces with. In the case of the local, time-fractional
wave equation on a bounded domain we refer to \cite{AKOZ,APZ-5,R-S-2010} and
references therein.

\subsection{System}

Introducing the dimensionless quantities%
\begin{gather*}
\bar{x}=x\left( \left( \tau _{\varepsilon }\right) ^{\frac{2}{\alpha }}\frac{%
\rho }{E}\right) ^{-\frac{1}{1+\beta }},\;\;\;\;\bar{t}=t\left( \tau
_{\varepsilon }\right) ^{-\frac{1}{\alpha }},\;\;\;\;\bar{u}=u\left( \left(
\tau _{\varepsilon }\right) ^{\frac{2}{\alpha }}\frac{\rho }{E}\right) ^{-%
\frac{1}{1+\beta }}, \\
\bar{\sigma}=\frac{\sigma }{E}\left( \left( \tau _{\varepsilon }\right) ^{%
\frac{2}{\alpha }}\frac{\rho }{E}\right) ^{-\frac{1-\beta }{1+\beta }%
},\;\;\;\;\bar{\varepsilon}=\varepsilon \left( \left( \tau _{\varepsilon
}\right) ^{\frac{2}{\alpha }}\frac{\rho }{E}\right) ^{-\frac{1-\beta }{%
1+\beta }},\;\;\;\;\tau =\frac{\tau _{\sigma }}{\tau _{\varepsilon }}, \\
\bar{u}_{0}=u_{0}\left( \left( \tau _{\varepsilon }\right) ^{\frac{2}{\alpha
}}\frac{\rho }{E}\right) ^{-\frac{1}{1+\beta }},\;\;\;\;\bar{v}%
_{0}=v_{0}\left( \tau _{\varepsilon }\right) ^{\frac{1}{\alpha }}\left(
\left( \tau _{\varepsilon }\right) ^{\frac{2}{\alpha }}\frac{\rho }{E}%
\right) ^{-\frac{1}{1+\beta }}
\end{gather*}%
in (\ref{eq:system1-d}) - (\ref{eq:system3-d}) and omitting bar we obtain
\begin{gather}
{\partial _{x}}\sigma (x,t)={\partial _{t}^{2}}u(x,t),\;\;x\in \mathbb{R}%
,\;t>0,  \label{eq:system1} \\
\sigma (x,t)+\tau \,{}_{0}^{C}\mathrm{D}_{t}^{\alpha }\sigma
(x,t)=\varepsilon (x,t)+{}_{0}^{C}\mathrm{D}_{t}^{\alpha }\varepsilon
(x,t),\;\;x\in \mathbb{R},\;t>0,  \label{eq:system2} \\
\varepsilon (x,t)=\mathcal{E}_{x}^{\beta }u(x,t),\;\;x\in \mathbb{R},\;t>0.
\label{eq:system3}
\end{gather}

System (\ref{eq:system1}) - (\ref{eq:system3}) can be reduced to the
space-time fractional Zener wave equation
\begin{equation}
\partial _{t}^{2}u(x,t)=L_{t}^{\alpha }\partial _{x}\mathcal{E}_{x}^{\beta
}u(x,t),\text{\ \ }x\in \mathbb{R},\;t>0,  \label{eq:fzwe}
\end{equation}%
where $L_{t}^{\alpha }$ is a linear operator (of convolution type) given by
\begin{equation}
L_{t}^{\alpha }=\mathcal{L}^{-1}\left[ \frac{1+s^{\alpha }}{1+\tau s^{\alpha
}}\right] \ast _{t}=\left( \frac{1}{\tau }\delta (t)+\left( \frac{1}{\tau }%
-1\right) e_{\alpha }^{\prime }(t)\right) \ast _{t},\;\;t>0,  \label{eq:L}
\end{equation}%
and $\mathcal{L}^{-1}$ denotes the inverse Laplace transform, see Appendix %
\ref{sec:math prelim}. The dimensionless quantities give that initial and
boundary conditions, (\ref{eq:bc-d}) and (\ref{eq:ic-d}), for the space-time
fractional Zener wave equation (\ref{eq:fzwe}) again become%
\begin{gather}
u(x,0)=u_{0}(x),\;\;\;\;{\partial _{t}}u(x,0)=v_{0}(x),\;\;\;\;\sigma
(x,0)=0,\;\;\;\;\varepsilon (x,0)=0,\;\;x\in \mathbb{R},  \label{eq:ic} \\
\lim_{x\rightarrow \pm \infty }u(x,t)=0,\;\;\;\;\lim_{x\rightarrow \pm
\infty }\sigma (x,t)=0,\;\;t>0.  \label{eq:bc}
\end{gather}

The procedure of obtaining (\ref{eq:fzwe}) is as follows. Applying the
Laplace transform to (\ref{eq:system2}) with respect to time variable $t$,
one obtains
\begin{equation*}
(1+\tau s^{\alpha })\tilde{\sigma}(x,s)=(1+s^{\alpha })\tilde{\varepsilon}%
(x,s),\;\;x\in
\mathbb{R}
,\;\func{Re}s>0.
\end{equation*}%
The inverse Laplace transform, since $\mathcal{L}^{-1}\left[ \frac{%
1+s^{\alpha }}{1+\tau s^{\alpha }}\right] $ is well-defined element in $%
\mathcal{S}_{+}^{\prime }$ (see \cite{Oparnica02}), gives
\begin{equation}
\sigma =\mathcal{L}^{-1}\left[ \frac{1+s^{\alpha }}{1+\tau s^{\alpha }}%
\right] \ast _{t}\varepsilon .  \label{sigma-epsilon}
\end{equation}%
Setting $L_{t}^{\alpha }=\mathcal{L}^{-1}\left[ \frac{1+s^{\alpha }}{1+\tau
s^{\alpha }}\right] \ast _{t},$ inserting $\varepsilon $, given by (\ref%
{eq:system3}), into (\ref{sigma-epsilon}) and then inserting obtained $%
\sigma $ into (\ref{eq:system1}), we obtain (\ref{eq:fzwe}). Note that $%
L_{t}^{\alpha }=\mathcal{L}^{-1}\left[ \frac{1+s^{\alpha }}{1+\tau s^{\alpha
}}\right] \ast _{t}$ can be explicitly expressed via the Mittag-Leffler
function. Recall, for the Mittag-Leffler function $e_{\alpha }$, defined by
\begin{equation}
e_{\alpha }(t)=E_{\alpha }\left( -\frac{t^{\alpha }}{\tau }\right)
,\;\;t>0,\;\alpha \in \left( 0,1\right) ,  \label{mlf}
\end{equation}%
where $E_{\alpha }\left( z\right) =\sum_{k=0}^{\infty }\frac{z^{k}}{\Gamma
(\alpha k+1)},$ $z\in
\mathbb{C}
,$ we have that $e_{\alpha }\in C^{\infty }((0,\infty ))\cap C([0,\infty ))$%
, $e_{\alpha }^{\prime }\in C^{\infty }((0,\infty ))\cap
L_{loc}^{1}([0,\infty )),$ where $e_{\alpha }^{\prime }(t)=\frac{\mathrm{d}}{%
\mathrm{d}t}e_{\alpha }(t),\;t>0,$ and%
\begin{equation*}
\mathcal{L}[e_{\alpha }(t)](s)=\frac{s^{\alpha -1}}{s^{\alpha }+\frac{1}{%
\tau }},\;\;\func{Re}s>0,
\end{equation*}%
cf.\ \cite{m-g-00}. Therefore,
\begin{equation*}
\mathcal{L}^{-1}\left[ \frac{1+s^{\alpha }}{1+\tau s^{\alpha }}\right] (t)=%
\mathcal{L}^{-1}\left[ 1+\frac{(1-\tau )s^{\alpha }}{\tau (s^{\alpha }+\frac{%
1}{\tau })}\right] (t)=\frac{1}{\tau }\delta (t)+\left( \frac{1}{\tau }%
-1\right) e_{\alpha }^{\prime }\left( t\right) ,\;\;t>0
\end{equation*}%
and thus we obtain $L_{t}^{\alpha }$ as given by (\ref{eq:L}).

For $\alpha =0$ and $\beta =1$, i.e., when the Hooke law and the classical
strain measure are used, equation (\ref{eq:fzwe}) is the classical wave
equation%
\begin{equation*}
\partial _{t}^{2}u=c^{2}\,\partial _{x}^{2}u,\;\;\text{with}\;\;c=\sqrt{%
\frac{2}{1+\tau }}.
\end{equation*}%
Therefore, system (\ref{eq:system1}) - (\ref{eq:system3}), or equivalently (%
\ref{eq:fzwe}), generalize the classical wave equation. We collect other
special cases of (\ref{eq:fzwe}) in following remark.

\begin{remark}
\label{cases-rem}Generalizations of the classical wave equation, given by
system (\ref{eq:system1}) - (\ref{eq:system3}), or (\ref{eq:fzwe}), are
distinguished and classified according to parameter $\beta $ as follows.

\begin{enumerate}
\item[$\left( i\right) $] Case $\beta =0.$ We obtain the non-propagating
disturbance if $v_{0}=0$. Namely, for $\beta =0,$ we obtain $\varepsilon =0,$
due to (\ref{eq:system3}) and the property of the symmetrized fractional
derivative that $\mathcal{E}_{x}^{0}u=0$, see Appendix \ref{sec:math prelim}%
. This and (\ref{eq:system2}), imply $\sigma =0,$ so that from (\ref%
{eq:system1}), (\ref{eq:ic}), and (\ref{eq:bc}) one obtains%
\begin{equation}
u\left( x,t\right) =u_{0}\left( x\right) +v_{0}\left( x\right) t,\;\;x\in
\mathbb{R},\;t\geq 0.  \label{non-prop}
\end{equation}%
Note, for $v_{0}=0,$ we have $u\left( x,t\right) =u_{0}\left( x\right) ,$ $%
x\in \mathbb{R},$ $t\geq 0.$

\item[$\left( ii\right) $] Case $\beta \in (0,1).$ For $\alpha =0$ we obtain
the space-fractional wave equation
\begin{equation}
\partial _{t}^{2}u(x,t)=c^{2}\,\partial _{x}\mathcal{E}_{x}^{\beta
}u(x,t),\;\;c=\sqrt{\frac{2}{1+\tau }},\;x\in \mathbb{R},\;t>0,  \label{sfwe}
\end{equation}%
studied in \cite{A-S-09}. Case $\alpha \in \left( 0,1\right) ,$ according to
authors' knowledge, have not been studied in the literature, so it is the
subject of analysis presented in this work. For $\alpha =1,$ (\ref{eq:fzwe})
becomes the space-fractional Zener wave equation
\begin{equation}
\partial _{t}^{2}u(x,t)=L_{t}^{1}\partial _{x}\mathcal{E}_{x}^{\beta
}u(x,t),\;\;x\in \mathbb{R},\;t>0.  \label{sfzwe}
\end{equation}%
For all $\alpha \in \left[ 0,1\right] ,$ when $\beta $ tends to zero,
solution to system (\ref{eq:system1}) - (\ref{eq:system3}), (\ref{eq:ic}), (%
\ref{eq:bc}) tends to (\ref{non-prop}), see Section \ref{cases}. This
suggests that the parameter $\beta $ measures the resistance of the material
to the propagation of initial disturbance.

\item[$\left( iii\right) $] In the case when $\beta =1,$ $\alpha \in (0,1),$
equation (\ref{eq:fzwe}) reduces to the time-fractional Zener wave equation%
\begin{equation}
\partial _{t}^{2}u(x,t)=L_{t}^{\alpha }\partial _{x}^{2}u(x,t),\;\;x\in
\mathbb{R},\;t>0,  \label{tfzwe}
\end{equation}%
studied in \cite{KOZ10,NH}. For $\alpha =0$, as already mentioned above, we
obtain the classical wave equation and for $\alpha =1$ Zener wave equation%
\begin{equation*}
\partial _{t}^{2}u(x,t)=L_{t}^{1}\partial _{x}^{2}u(x,t),\;\;x\in \mathbb{R}%
,\;t>0.
\end{equation*}
\end{enumerate}
\end{remark}

\section{Cauchy problem (\protect\ref{eq:fzwe}), (\protect\ref{eq:ic}) \label%
{sec:fund sol}}

\subsection{Framework}

The framework for our analysis are the spaces of distributions: $\mathcal{S}%
^{\prime }(\mathbb{R})$ (or shortly $\mathcal{S}^{\prime }$) and $\mathcal{K}%
^{\prime }(\mathbb{R})$ (or $\mathcal{K}^{\prime }$) the duals of the
Schwartz space $\mathcal{S}(\mathbb{R})$ (or $\mathcal{S}$) and of the space
$\mathcal{K}(\mathbb{R})$ (or $\mathcal{K}$); $\mathcal{K}$ is the space of
smooth functions $\varphi $ with the property $\sup_{x\in \mathbb{R},\alpha
\leq m}\left\vert \varphi ^{(\alpha )}(x)\right\vert \mathrm{e}^{m\left\vert
x\right\vert }<\infty ,$ $m\in \mathbb{N}_{0}.$ The elements of $\mathcal{S}%
^{\prime },$ respectively of $\mathcal{K}^{\prime },$ are of the form $%
f=\sum_{\alpha =0}^{r}\Phi _{\alpha }^{(\alpha )},$ where $\Phi _{\alpha }$
are continuous functions on $%
\mathbb{R}
$ and $\left\vert \Phi _{\alpha }(t)\right\vert \leq C\left( 1+\left\vert
t\right\vert \right) ^{k_{0}},$ respectively $\left\vert \Phi _{\alpha
}(t)\right\vert \leq C\mathrm{e}^{k_{0}|t|},$ $\alpha \leq r,\;t\in \mathbb{R%
},$ for some $C>0$, $r\in \mathbb{N}_{0}$ and $k_{0}\in \mathbb{N}_{0}.$ The
space $\mathcal{S}_{+}^{\prime }$ ($\mathcal{K}_{+}^{\prime }$) is a
subspace of $\mathcal{S}^{\prime }$ ($\mathcal{K}^{\prime }$) consisting of
elements supported by $[0,\infty )$. The elements of $\mathcal{S}%
_{+}^{\prime },$ respectively of $\mathcal{K}_{+}^{\prime },$ are of the
form $f\left( t\right) =(\Phi (t)\left( 1+\left\vert t\right\vert \right)
^{k})^{(p)}$, respectively $f\left( t\right) =(\Phi (t)\mathrm{e}%
^{kt})^{(p)},$ $t\in \mathbb{R},$ where $\Phi $ is a continuous bounded
function such that $\Phi (t)=0,\;t\leq 0$. Note that $\mathcal{S}^{\prime }$
and $\mathcal{S}_{+}^{\prime }$ are subspaces of $\mathcal{K}^{\prime }$ and
$\mathcal{K}_{+}^{\prime },$ respectively. The elements of $\mathcal{K}%
_{+}^{\prime }$ have the Laplace transform, which are analytic functions in
the domain $\func{Re}s>s_{0}>0.$ We also recall that for the Lebesgue spaces
of integrable and bounded functions $L^{1}\left( \mathbf{%
\mathbb{R}
}\right) $ and $L^{\infty }\left( \mathbf{%
\mathbb{R}
}\right) ,$ $f\ast g\in L^{\infty }\left( \mathbf{%
\mathbb{R}
}\right) $ if $f\in L^{1}\left( \mathbf{%
\mathbb{R}
}\right) $ and $g\in L^{\infty }\left( \mathbf{%
\mathbb{R}
}\right) $.

We shall apply the Fourier transform with respect to $x$ and the Laplace
transform with respect to $t.$ Actually, we shall consider the distributions
within the space $\mathcal{S}^{\prime }\otimes \mathcal{K}_{+}^{\prime },$
which is the subspace of $\mathcal{K}^{\prime }\left(
\mathbb{R}
^{2}\right) ,$ consisting of distributions having support in $%
\mathbb{R}
\times \left[ 0,\infty \right) .$ For the background of tensor product, we
refer to \cite{Treves-PDEs}. We shall obtain the solution $u$\ as an element
of $C\left(
\mathbb{R}
\right) \cap L^{\infty }\left(
\mathbb{R}
\right) $\ for fixed $\varphi \in K,$\ i.e., $\left\langle u\left(
x,t\right) ,\varphi \left( t\right) \right\rangle \in C\left(
\mathbb{R}
\right) \cap L^{\infty }\left(
\mathbb{R}
\right) $\ and elements of $K_{+}^{\prime }$\ for fixed $\psi \in K,$\ i.e.,
$\left\langle u\left( x,t\right) ,\psi \left( x\right) \right\rangle \in
K_{+}^{\prime }.$

\subsection{Existence and uniqueness of a generalized solution \label{eu}}

We consider the existence and uniqueness of the solution to the Cauchy
problem (\ref{eq:fzwe}), (\ref{eq:ic}), (\ref{eq:bc}). If $u_{0}\in C^{1}(%
\mathbb{R})$ and $v_{0}\in C(\mathbb{R})$, then the classical solution to
the Cauchy problem (\ref{eq:fzwe}), (\ref{eq:ic}), (\ref{eq:bc}) is a
function $u(x,t)$ of class $C^{2}$ for $t>0$, of class $C^{1}$ for $t\geq 0$%
, which satisfies equation (\ref{eq:fzwe}) for $t>0$ and initial conditions (%
\ref{eq:ic}) when $t=0,$ as well as the boundary conditions (\ref{eq:bc}).
If the function $u$ is continued by zero for $t<0$, then putting
\begin{equation*}
u(x,t)=\mathcal{U}(x,t)H\left( t\right) ,\;\;x,t\in
\mathbb{R}
,
\end{equation*}%
we obtain%
\begin{equation}
\partial _{t}^{2}\mathcal{U}(x,t)=L_{t}^{\alpha }\partial _{x}\mathcal{E}%
_{x}^{\beta }\mathcal{U}(x,t)+u_{0}(x)\delta ^{\prime }(t)+v_{0}(x)\delta
(t),\;\;\text{in}\;\;\mathcal{K}^{\prime }(\mathbb{R}^{2}).
\label{eq:fzwe-distr}
\end{equation}

The main theorem is the following one.

\begin{theorem}
\label{th:glavna}Let $\alpha \in \lbrack 0,1)$, $\beta \in \lbrack 0,1)$, $%
\tau \in (0,1)$ and let $u_{0},v_{0}\in L^{1}\left(
\mathbb{R}
\right) $. Then there exists a unique generalized solution $u\in \mathcal{K}%
^{\prime }\left(
\mathbb{R}
^{2}\right) ,$ $\limfunc{supp}u\subset
\mathbb{R}
\times \left[ 0,\infty \right) ,$ to the space-time fractional Zener wave
equation (\ref{eq:fzwe}), with initial (\ref{eq:ic}) and boundary data (\ref%
{eq:bc}).

More precisely, $u$ is of the form%
\begin{equation}
u(x,t)=\frac{1}{2\pi ^{2}}\left( \delta ^{\prime }\left( t\right)
u_{0}\left( x\right) +\delta \left( t\right) v_{0}\left( x\right) \right)
\ast _{x,t}P\left( x,t\right) ,\;\;x\in
\mathbb{R}
,\;t>0,  \label{u-od-x-i-t}
\end{equation}%
where%
\begin{equation*}
P\left( x,t\right) =I\left( x,t\right) -\left( \frac{\partial }{\partial t}%
J_{1}\left( x,t\right) +\frac{\partial ^{2}}{\partial t^{2}}J_{2}\left(
x,t\right) \right) \mathrm{e}^{s_{0}t},\;\;x\in
\mathbb{R}
,\;t>0,
\end{equation*}%
with%
\begin{equation*}
J_{1}=\mathrm{i\,}\left( J_{1}^{+}-J_{1}^{-}\right)
,\;\;\;\;J_{2}=J_{2}^{+}+J_{2}^{-}
\end{equation*}%
and ($x\in
\mathbb{R}
,$ $t>0$)%
\begin{eqnarray*}
I\left( x,t\right) &=&\int_{-p_{0}}^{p_{0}}\int_{0}^{\infty }\frac{\cos
(\rho x)\mathrm{e}^{s_{0}t}\mathrm{e}^{\mathrm{i}pt}}{\left[ s^{2}+\frac{%
1+s^{\alpha }}{1+\tau s^{\alpha }}\rho ^{1+\beta }\sin \frac{\beta \pi }{2}%
\right] _{s=s_{0}+\mathrm{i}p}}\mathrm{d}\rho \,\mathrm{d}p, \\
J_{1}^{+}\left( x,t\right) &=&\int_{p_{0}}^{\infty }\int_{0}^{1}\frac{\cos
(\rho x)\mathrm{e}^{\mathrm{i}pt}}{p\left[ s^{2}+\frac{1+s^{\alpha }}{1+\tau
s^{\alpha }}\rho ^{1+\beta }\sin \frac{\beta \pi }{2}\right] _{s=s_{0}+%
\mathrm{i}p}}\mathrm{d}\rho \,\mathrm{d}p, \\
J_{1}^{-}\left( x,t\right) &=&\int_{p_{0}}^{\infty }\int_{0}^{1}\frac{\cos
(\rho x)\mathrm{e}^{-\mathrm{i}qt}}{q\left[ s^{2}+\frac{1+s^{\alpha }}{%
1+\tau s^{\alpha }}\rho ^{1+\beta }\sin \frac{\beta \pi }{2}\right]
_{s=s_{0}-\mathrm{i}q}}\mathrm{d}\rho \,\mathrm{d}q, \\
J_{2}^{+}\left( x,t\right) &=&\int_{p_{0}}^{\infty }\int_{1}^{\infty }\frac{%
\cos (\rho x)\mathrm{e}^{\mathrm{i}pt}}{p^{2}\left[ s^{2}+\frac{1+s^{\alpha }%
}{1+\tau s^{\alpha }}\rho ^{1+\beta }\sin \frac{\beta \pi }{2}\right]
_{s=s_{0}+\mathrm{i}p}}\mathrm{d}\rho \,\mathrm{d}p, \\
J_{2}^{-}\left( x,t\right) &=&\int_{p_{0}}^{\infty }\int_{1}^{\infty }\frac{%
\cos (\rho x)\mathrm{e}^{-\mathrm{i}qt}}{q^{2}\left[ s^{2}+\frac{1+s^{\alpha
}}{1+\tau s^{\alpha }}\rho ^{1+\beta }\sin \frac{\beta \pi }{2}\right]
_{s=s_{0}-\mathrm{i}q}}\mathrm{d}\rho \,\mathrm{d}q.
\end{eqnarray*}

Functions $I,$ $J_{1}^{+},$ $J_{1}^{-},$ $J_{2}^{+}$ and $J_{2}^{-}$ are
bounded and continuous functions with respect to $x$ and continuous
exponentially bounded functions with respect to $t.$
\end{theorem}

\begin{proof}
The plan of the proof is to solve (\ref{eq:fzwe-distr}) with the assumption
that $u_{0}$ and $v_{0}$ are compactly supported smooth functions, i.e.,
elements of $C_{0}^{\infty }\left(
\mathbb{R}
\right) $. Namely, if sequences $\left\{ u_{0n}\right\} _{n\in
\mathbb{N}
},\left\{ v_{0n}\right\} _{n\in
\mathbb{N}
}\in C_{0}^{\infty }\left(
\mathbb{R}
\right) $ are such that $u_{0n}\rightarrow u_{0}$ and $v_{0n}\rightarrow
v_{0}$ in $L^{1}\left( \mathbf{%
\mathbb{R}
}\right) $ as $n\rightarrow \infty ,$ then (\ref{u-od-x-i-t}) is understood
as%
\begin{eqnarray*}
&&\left( \delta ^{\prime }\left( t\right) u_{0}\left( x\right) +\delta
\left( t\right) v_{0}\left( x\right) \right) \ast _{x,t}P\left( x,t\right) \\
&&\qquad \qquad \qquad =\lim_{n\rightarrow \infty }\left( \left( \delta
^{\prime }\left( t\right) u_{0n}\left( x\right) +\delta \left( t\right)
v_{0n}\left( x\right) \right) \ast _{x,t}P\left( x,t\right) \right) \;\;%
\text{in}\;\;L^{\infty }\left( \mathbf{%
\mathbb{R}
}\right) .
\end{eqnarray*}%
Hence, (\ref{u-od-x-i-t}), with $u_{0},v_{0}\in L^{1}\left(
\mathbb{R}
\right) ,$ is a solution to (\ref{eq:fzwe-distr}). In the sequel, we assume
that $u_{0},v_{0}\in C_{0}^{\infty }\left(
\mathbb{R}
\right) .$ This enables us to use the exchange formula.

Formally applying the Laplace transform to (\ref{eq:fzwe}) with respect to $%
t $, with the initial conditions (\ref{eq:ic}) taken into account, we obtain%
\begin{equation}
\partial _{x}\mathcal{E}_{x}^{\beta }\tilde{u}(x,s)-s^{2}\frac{1+\tau
s^{\alpha }}{1+s^{\alpha }}\tilde{u}(x,s)=-\frac{1+\tau s^{\alpha }}{%
1+s^{\alpha }}(su_{0}(x)+v_{0}(x)),\;\;x\in
\mathbb{R}
,\;\func{Re}s>s_{0},  \label{ODE_Sim-1-cstd}
\end{equation}%
for suitably chosen $s_{0}>0,$ where $\tilde{u}$ is an analytic function
with respect to $s.$ Equation (\ref{ODE_Sim-1-cstd}) is of the type%
\begin{equation}
\partial _{x}\mathcal{E}_{x}^{\beta }u\left( x\right) -\omega \,u\left(
x\right) =-\nu \,u_{0}\left( x\right) -\mu \,v_{0}\left( x\right) ,\;\;x\in
\mathbb{R}
,  \label{ODE_Sim-cstd}
\end{equation}%
where
\begin{equation}
\omega =\omega (s)=s^{2}\frac{1+\tau s^{\alpha }}{1+s^{\alpha }},\;\;\nu
=\nu (s)=s\frac{1+\tau s^{\alpha }}{1+s^{\alpha }},\;\;\mu =\mu \left(
s\right) =\frac{1+\tau s^{\alpha }}{1+s^{\alpha }},\;\;\func{Re}s>s_{0}.
\label{omega,ni,mi}
\end{equation}%
We have shown in \cite[Theorem 4.2]{KOZ10} that $\omega (s)\in
\mathbb{C}
\setminus (-\infty ,0]$ for $\func{Re}s>0.$ For fixed $s,$ $\func{Re}%
s>s_{0}, $ the unique solution $u\in C\left(
\mathbb{R}
\right) \cap L^{\infty }\left(
\mathbb{R}
\right) $ to (\ref{ODE_Sim-cstd}), given by%
\begin{equation}
u\left( x\right) =\frac{1}{\pi }\left( \nu \,u_{0}\left( x\right) +\mu
\,v_{0}\left( x\right) \right) \ast _{x}\int_{0}^{\infty }\frac{1}{\rho
^{1+\beta }\sin \frac{\beta \pi }{2}+\omega }\cos (\rho x)\mathrm{d}\rho
,\;\;x\in
\mathbb{R}
,  \label{ODE_Sim-2-cstd}
\end{equation}%
is obtained as the inverse Fourier transform of
\begin{equation*}
\hat{u}\left( \xi \right) =\frac{\nu \,\hat{u}_{0}\left( \xi \right) +\mu \,%
\hat{v}_{0}\left( \xi \right) }{|\xi |^{1+\beta }\sin \frac{\beta \pi }{2}%
+\omega },\;\;\xi \in
\mathbb{R}
.
\end{equation*}%
The previous expression, with $\omega ,$ $\nu $ and $\mu $ given by (\ref%
{omega,ni,mi}), takes the form%
\begin{equation}
\widehat{\tilde{u}}\left( \xi ,s\right) =\frac{s\hat{u}_{0}\left( \xi
\right) +\hat{v}_{0}\left( \xi \right) }{s^{2}+\frac{1+s^{\alpha }}{1+\tau
s^{\alpha }}|\xi |^{1+\beta }\sin \frac{\beta \pi }{2}},\;\;\xi \in
\mathbb{R}
,\;\func{Re}s>s_{0}.  \label{u-tilda,het}
\end{equation}%
In fact, we have that $x\mapsto \int_{0}^{\infty }\frac{1}{\rho ^{1+\beta
}\sin \frac{\beta \pi }{2}+\omega }\cos (\rho x)d\rho $\ is a continuous
bounded function ($C\left(
\mathbb{R}
\right) \cap L^{\infty }\left(
\mathbb{R}
\right) $) for fixed $\omega =\omega \left( s\right) $. After this function
is convoluted with $\nu \,u_{0}\left( x\right) +\mu \,v_{0}\left( x\right) ,$%
\ where $u_{0},v_{0}\in L^{1}\left(
\mathbb{R}
\right) ,$\ one obtains the function that belongs to $C\left(
\mathbb{R}
\right) \cap L^{\infty }\left(
\mathbb{R}
\right) .$ Therefore, by (\ref{ODE_Sim-2-cstd}), we have the solution to (%
\ref{ODE_Sim-1-cstd}) in the form%
\begin{align}
\tilde{u}(x,s)& =\frac{1}{\pi }\frac{1+\tau s^{\alpha }}{1+s^{\alpha }}%
\left( su_{0}\left( x\right) +v_{0}\left( x\right) \right) \ast
_{x}\int_{0}^{\infty }\frac{1}{\rho ^{1+\beta }\sin \frac{\beta \pi }{2}%
+s^{2}\frac{1+\tau s^{\alpha }}{1+s^{\alpha }}}\cos (\rho x)\mathrm{d}\rho
\notag \\
& =\frac{1}{\pi }\left( su_{0}\left( x\right) +v_{0}\left( x\right) \right)
\ast _{x}\int_{0}^{\infty }\frac{1}{s^{2}+\frac{1+s^{\alpha }}{1+\tau
s^{\alpha }}\rho ^{1+\beta }\sin \frac{\beta \pi }{2}}\cos (\rho x)\mathrm{d}%
\rho ,\;\;x\in
\mathbb{R}
,\;\func{Re}s>s_{0}.  \label{eq:LTT-cstd}
\end{align}

Thus, the justification of the previously presented procedure is based on
the analysis of the inverse Laplace transform. Formally, when applied to (%
\ref{eq:LTT-cstd}) the inverse Laplace transform gives%
\begin{equation}
u(x,t)=\frac{1}{2\pi ^{2}}\left( \delta ^{\prime }\left( t\right)
u_{0}\left( x\right) +\delta \left( t\right) v_{0}\left( x\right) \right)
\ast _{x,t}P\left( x,t\right) ,\;\;x\in
\mathbb{R}
,\;t>0,  \label{T-cstd}
\end{equation}%
where%
\begin{eqnarray}
P\left( x,t\right) &=&2\pi \mathcal{L}^{-1}\left[ \int_{0}^{\infty }\frac{%
\cos (\rho x)}{s^{2}+\frac{1+s^{\alpha }}{1+\tau s^{\alpha }}\rho ^{1+\beta
}\sin \frac{\beta \pi }{2}}\mathrm{d}\rho \right] \left( x,t\right)
\label{P-cstd} \\
&=&-\mathrm{i}\int_{s_{0}-\mathrm{i}\infty }^{s_{0}+\mathrm{i}\infty
}\int_{0}^{\infty }\frac{\cos (\rho x)\mathrm{e}^{st}}{s^{2}+\frac{%
1+s^{\alpha }}{1+\tau s^{\alpha }}\rho ^{1+\beta }\sin \frac{\beta \pi }{2}}%
\mathrm{d}\rho \,\mathrm{d}s,\;\;x\in
\mathbb{R}
,\;t>0.  \label{Q-cstd}
\end{eqnarray}%
Consider the divergent integral (\ref{Q-cstd}). We introduce the
parametrization $s=s_{0}+\mathrm{i}p,$ $p\in
\mathbb{R}
,$ in (\ref{Q-cstd}), so that%
\begin{eqnarray*}
P\left( x,t\right) &=&\int_{-\infty }^{\infty }\int_{0}^{\infty }\frac{\cos
(\rho x)\mathrm{e}^{s_{0}t}\mathrm{e}^{\mathrm{i}pt}}{\left[ s^{2}+\frac{%
1+s^{\alpha }}{1+\tau s^{\alpha }}\rho ^{1+\beta }\sin \frac{\beta \pi }{2}%
\right] _{s=s_{0}+\mathrm{i}p}}\mathrm{d}\rho \,\mathrm{d}p \\
&=&I\left( x,t\right) +I^{+}\left( x,t\right) +I^{-}\left( x,t\right)
,\;\;x\in
\mathbb{R}
,\;t>0,
\end{eqnarray*}%
where $I,$ $I^{+}$ and $I^{-}$ are given below.

The integral%
\begin{equation*}
I\left( x,t\right) =\int_{-p_{0}}^{p_{0}}\int_{0}^{\infty }\frac{\cos (\rho
x)\mathrm{e}^{s_{0}t}\mathrm{e}^{\mathrm{i}pt}}{\left[ s^{2}+\frac{%
1+s^{\alpha }}{1+\tau s^{\alpha }}\rho ^{1+\beta }\sin \frac{\beta \pi }{2}%
\right] _{s=s_{0}+\mathrm{i}p}}\mathrm{d}\rho \,\mathrm{d}p,\;\;x\in
\mathbb{R}
,\;t>0,
\end{equation*}%
is absolutely convergent, since
\begin{equation}
\left\vert I\left( x,t\right) \right\vert \leq \mathrm{e}^{s_{0}t}%
\int_{-p_{0}}^{p_{0}}\int_{0}^{\infty }\frac{1}{\func{Re}\left( \left[ s^{2}+%
\frac{1+s^{\alpha }}{1+\tau s^{\alpha }}\rho ^{1+\beta }\sin \frac{\beta \pi
}{2}\right] _{s=s_{0}+\mathrm{i}p}\right) }\mathrm{d}\rho \,\mathrm{d}%
p<\infty ,\;\;x\in
\mathbb{R}
,\;t>0.  \label{i-moduo-cstd}
\end{equation}%
In (\ref{i-moduo-cstd}), $p_{0}$ is chosen so that%
\begin{eqnarray*}
&&\func{Re}\left( \left[ s^{2}+\frac{1+s^{\alpha }}{1+\tau s^{\alpha }}\rho
^{1+\beta }\sin \frac{\beta \pi }{2}\right] _{s=s_{0}+\mathrm{i}p}\right) \\
&&\qquad \qquad =r^{2}\cos (2\varphi )+\frac{1+(1+\tau )r^{\alpha }\cos
(\alpha \varphi )+\tau r^{2\alpha }}{1+2\tau r^{\alpha }\cos (\alpha \varphi
)+\tau ^{2}r^{2\alpha }}\rho ^{1+\beta }\sin \frac{\beta \pi }{2}>0,
\end{eqnarray*}%
with $r=\sqrt{s_{0}^{2}+p^{2}}$ and $\tan \varphi =\frac{p}{s_{0}}.$ Note
that for $p=0$ and $\rho =0$ the integrand is well-defined, due to $s_{0}>0.$
Thus, the integral $I$ exists and belongs to $C\left(
\mathbb{R}
\right) \cap L^{\infty }\left(
\mathbb{R}
\right) $ with respect to $x$ and it is a continuous exponentially bounded
function with respect to $t.$ Next, we consider%
\begin{eqnarray}
I^{+}\left( x,t\right) &=&\mathrm{e}^{s_{0}t}\int_{p_{0}}^{\infty
}\int_{0}^{\infty }\frac{\cos (\rho x)\mathrm{e}^{\mathrm{i}pt}}{\left[
s^{2}+\frac{1+s^{\alpha }}{1+\tau s^{\alpha }}\rho ^{1+\beta }\sin \frac{%
\beta \pi }{2}\right] _{s=s_{0}+\mathrm{i}p}}\mathrm{d}\rho \,\mathrm{d}p
\notag \\
&=&-\mathrm{e}^{s_{0}t}\left( \mathrm{i\,}\frac{\partial }{\partial t}%
J_{1}^{+}\left( x,t\right) +\frac{\partial ^{2}}{\partial t^{2}}%
J_{2}^{+}\left( x,t\right) \right) ,\;\;x\in
\mathbb{R}
,\;t>0,  \label{iplus-cstd}
\end{eqnarray}%
where $J_{1}^{+}$ and $J_{2}^{+}$ are defined below. Setting%
\begin{equation*}
J_{1}^{+}\left( x,t\right) =\int_{p_{0}}^{\infty }\int_{0}^{1}\frac{\cos
(\rho x)\mathrm{e}^{\mathrm{i}pt}}{p\left[ s^{2}+\frac{1+s^{\alpha }}{1+\tau
s^{\alpha }}\rho ^{1+\beta }\sin \frac{\beta \pi }{2}\right] _{s=s_{0}+%
\mathrm{i}p}}\mathrm{d}\rho \,\mathrm{d}p,\;\;x\in
\mathbb{R}
,\;t>0,
\end{equation*}%
we have that the integral $J_{1}^{+}$ exists and belongs to $C\left(
\mathbb{R}
\right) \cap L^{\infty }\left(
\mathbb{R}
\right) $ with respect to both $x$ and $t,$ since%
\begin{eqnarray*}
\left\vert J_{1}^{+}\left( x,t\right) \right\vert &\leq
&\int_{p_{0}}^{\infty }\int_{0}^{1}\frac{1}{p\func{Im}\left( \left[ s^{2}+%
\frac{1+s^{\alpha }}{1+\tau s^{\alpha }}\rho ^{1+\beta }\sin \frac{\beta \pi
}{2}\right] _{s=s_{0}+\mathrm{i}p}\right) }\mathrm{d}\rho \,\mathrm{d}p \\
&\leq &\frac{1}{2s_{0}}\int_{p_{0}}^{\infty }\int_{0}^{1}\frac{1}{p^{2}}%
\mathrm{d}\rho \,\mathrm{d}p<\infty ,\;\;x\in
\mathbb{R}
,\;t>0,
\end{eqnarray*}%
where we used%
\begin{eqnarray}
&&\func{Im}\left( \left[ s^{2}+\frac{1+s^{\alpha }}{1+\tau s^{\alpha }}\rho
^{1+\beta }\sin \frac{\beta \pi }{2}\right] _{s=s_{0}+\mathrm{i}p}\right)
\notag \\
&&\qquad \qquad =r^{2}\sin (2\varphi )+(1-\tau )\frac{r^{\alpha }\sin
(\alpha \varphi )}{1+2\tau r^{\alpha }\cos (\alpha \varphi )+\tau
^{2}r^{2\alpha }}\rho ^{1+\beta }\sin \frac{\beta \pi }{2}>0  \notag \\
&&\qquad \qquad \sim r^{2}\sin (2\varphi )+\frac{1-\tau }{\tau ^{2}}\frac{1}{%
r^{\alpha }}\rho ^{1+\beta }\sin (\alpha \varphi )\sin \frac{\beta \pi }{2}%
,\;\;\text{as}\;\;r\rightarrow \infty  \notag \\
&&\qquad \qquad \sim 2s_{0}p+\frac{1-\tau }{\tau ^{2}}\frac{1}{p^{\alpha }}%
\rho ^{1+\beta }\sin \frac{\alpha \pi }{2}\sin \frac{\beta \pi }{2},\;\;%
\text{as}\;\;p\rightarrow \infty .  \label{im-asimpt}
\end{eqnarray}%
In obtaining (\ref{im-asimpt}) we used: $r^{2}\sin (2\varphi )=r^{2}\frac{%
2\tan \varphi }{1+\tan ^{2}\varphi }=2s_{0}p,$ as well as $\varphi \sim
\frac{\pi }{2}$ and $r\sim p,$ as $p\rightarrow \infty .$ We put
\begin{equation*}
J_{2}^{+}\left( x,t\right) =\int_{p_{0}}^{\infty }\int_{1}^{\infty }\frac{%
\cos (\rho x)\mathrm{e}^{\mathrm{i}pt}}{p^{2}\left[ s^{2}+\frac{1+s^{\alpha }%
}{1+\tau s^{\alpha }}\rho ^{1+\beta }\sin \frac{\beta \pi }{2}\right]
_{s=s_{0}+\mathrm{i}p}}\mathrm{d}\rho \,\mathrm{d}p,\;\;x\in
\mathbb{R}
,\;t>0,
\end{equation*}%
which, by (\ref{im-asimpt}) and the Fubini theorem gives%
\begin{eqnarray*}
\left\vert J_{2}^{+}\left( x,t\right) \right\vert &\leq
&\int_{p_{0}}^{\infty }\int_{1}^{\infty }\frac{1}{p^{2}\func{Im}\left( \left[
s^{2}+\frac{1+s^{\alpha }}{1+\tau s^{\alpha }}\rho ^{1+\beta }\sin \frac{%
\beta \pi }{2}\right] _{s=s_{0}+\mathrm{i}p}\right) }\mathrm{d}\rho \,%
\mathrm{d}p \\
&\leq &\int_{p_{0}}^{\infty }\left( \int_{1}^{\infty }\frac{1}{2s_{0}p^{3}+%
\frac{1-\tau }{\tau ^{2}}p^{2-\alpha }\rho ^{1+\beta }\sin \frac{\alpha \pi
}{2}\sin \frac{\beta \pi }{2}}\mathrm{d}\rho \right) \mathrm{d}p \\
&\leq &\frac{\tau ^{2}}{1-\tau }\frac{1}{\sin \frac{\alpha \pi }{2}\sin
\frac{\beta \pi }{2}}\int_{p_{0}}^{\infty }\frac{1}{p^{2-\alpha }}\left(
\int_{1}^{\infty }\frac{1}{\rho ^{1+\beta }}\mathrm{d}\rho \right) \mathrm{d}%
p<\infty ,\;\;x\in
\mathbb{R}
,\;t>0.
\end{eqnarray*}%
By the same arguments as for $J_{1}^{+}$, we that the integral $J_{2}^{+}$
exists and belongs to $C\left(
\mathbb{R}
\right) \cap L^{\infty }\left(
\mathbb{R}
\right) $ with respect to both $x$ and $t.$ Thus, we have that $I^{+},$
given by (\ref{iplus-cstd}), belongs to $C\left(
\mathbb{R}
\right) \cap L^{\infty }\left(
\mathbb{R}
\right) $ with respect to $x,$ and it is a derivative of a continuous
exponentially bounded function with respect to $t.$ Similarly as for $I^{+}$%
, we prove the existence for
\begin{eqnarray}
I^{-}\left( x,t\right) &=&\int_{-\infty }^{-p_{0}}\int_{0}^{\infty }\frac{%
\cos (\rho x)\mathrm{e}^{s_{0}t}\mathrm{e}^{\mathrm{i}pt}}{\left[ s^{2}+%
\frac{1+s^{\alpha }}{1+\tau s^{\alpha }}\rho ^{1+\beta }\sin \frac{\beta \pi
}{2}\right] _{s=s_{0}+\mathrm{i}p}}\mathrm{d}\rho \,\mathrm{d}p  \notag \\
&=&\int_{p_{0}}^{\infty }\int_{0}^{\infty }\frac{\cos (\rho x)\mathrm{e}%
^{s_{0}t}\mathrm{e}^{-\mathrm{i}qt}}{\left[ s^{2}+\frac{1+s^{\alpha }}{%
1+\tau s^{\alpha }}\rho ^{1+\beta }\sin \frac{\beta \pi }{2}\right]
_{s=s_{0}-\mathrm{i}q}}\mathrm{d}\rho \,\mathrm{d}q  \notag \\
&=&\mathrm{e}^{s_{0}t}\left( \mathrm{i\,}\frac{\partial }{\partial t}%
J_{1}^{-}\left( x,t\right) -\frac{\partial ^{2}}{\partial t^{2}}%
J_{2}^{-}\left( x,t\right) \right) ,\;\;x\in
\mathbb{R}
,\;t>0.  \label{iminus-cstd}
\end{eqnarray}%
We have%
\begin{equation*}
J_{1}^{-}\left( x,t\right) =\int_{p_{0}}^{\infty }\int_{0}^{1}\frac{\cos
(\rho x)\mathrm{e}^{-\mathrm{i}qt}}{q\left[ s^{2}+\frac{1+s^{\alpha }}{%
1+\tau s^{\alpha }}\rho ^{1+\beta }\sin \frac{\beta \pi }{2}\right]
_{s=s_{0}-\mathrm{i}q}}\mathrm{d}\rho \,\mathrm{d}q,\;\;x\in
\mathbb{R}
,\;t>0,
\end{equation*}%
so that%
\begin{equation*}
\left\vert J_{1}^{-}\left( x,t\right) \right\vert \leq \int_{p_{0}}^{\infty
}\int_{0}^{1}\frac{1}{q\left\vert \func{Im}\left( \left[ s^{2}+\frac{%
1+s^{\alpha }}{1+\tau s^{\alpha }}\rho ^{1+\beta }\sin \frac{\beta \pi }{2}%
\right] _{s=s_{0}-\mathrm{i}q}\right) \right\vert }\mathrm{d}\rho \,\mathrm{d%
}q,\;\;x\in
\mathbb{R}
,\;t>0.
\end{equation*}%
The integral $J_{1}^{-}$ exists and belongs to $C\left(
\mathbb{R}
\right) \cap L^{\infty }\left(
\mathbb{R}
\right) $ with respect to both $x$ and $t,$ since%
\begin{equation*}
\left\vert J_{1}^{-}\left( x,t\right) \right\vert \leq \frac{1}{2s_{0}}%
\int_{p_{0}}^{\infty }\int_{0}^{1}\frac{1}{q^{2}}\mathrm{d}\rho \,\mathrm{d}%
q<\infty ,\;\;x\in
\mathbb{R}
,\;t>0,
\end{equation*}%
where $r=\sqrt{s_{0}^{2}+q^{2}},$ $\tan \varphi =-\frac{q}{s_{0}}$ and where
we used%
\begin{eqnarray*}
&&\func{Im}\left( \left[ s^{2}+\frac{1+s^{\alpha }}{1+\tau s^{\alpha }}\rho
^{1+\beta }\sin \frac{\beta \pi }{2}\right] _{s=s_{0}-\mathrm{i}q}\right) \\
&&\qquad \qquad =r^{2}\sin (2\varphi )+(1-\tau )\frac{r^{\alpha }\sin
(\alpha \varphi )}{1+2\tau r^{\alpha }\cos (\alpha \varphi )+\tau
^{2}r^{2\alpha }}\rho ^{1+\beta }\sin \frac{\beta \pi }{2} \\
&&\qquad \qquad \sim r^{2}\sin (2\varphi )+\frac{1-\tau }{\tau ^{2}}\frac{1}{%
r^{\alpha }}\rho ^{1+\beta }\sin (\alpha \varphi )\sin \frac{\beta \pi }{2}%
<0,\;\;\text{as}\;\;r\rightarrow \infty \\
&&\qquad \qquad \sim -\left( 2s_{0}q+\frac{1-\tau }{\tau ^{2}}\frac{1}{%
q^{\alpha }}\rho ^{1+\beta }\sin \frac{\alpha \pi }{2}\sin \frac{\beta \pi }{%
2}\right) ,\;\;\text{as}\;\;q\rightarrow \infty .
\end{eqnarray*}%
Consider
\begin{equation*}
J_{2}^{-}\left( x,t\right) =\int_{p_{0}}^{\infty }\int_{1}^{\infty }\frac{%
\cos (\rho x)\mathrm{e}^{-\mathrm{i}qt}}{q^{2}\left[ s^{2}+\frac{1+s^{\alpha
}}{1+\tau s^{\alpha }}\rho ^{1+\beta }\sin \frac{\beta \pi }{2}\right]
_{s=s_{0}-\mathrm{i}q}}\mathrm{d}\rho \,\mathrm{d}q,\;\;x\in
\mathbb{R}
,\;t>0.
\end{equation*}%
We have%
\begin{eqnarray*}
\left\vert J_{2}^{-}\left( x,t\right) \right\vert &\leq
&\int_{p_{0}}^{\infty }\int_{1}^{\infty }\frac{1}{q^{2}\left\vert \func{Im}%
\left( \left[ s^{2}+\frac{1+s^{\alpha }}{1+\tau s^{\alpha }}\rho ^{1+\beta
}\sin \frac{\beta \pi }{2}\right] _{s=s_{0}-\mathrm{i}q}\right) \right\vert }%
\mathrm{d}\rho \,\mathrm{d}q \\
&\leq &\int_{p_{0}}^{\infty }\left( \int_{1}^{\infty }\frac{1}{2s_{0}q^{3}+%
\frac{1-\tau }{\tau ^{2}}q^{2-\alpha }\rho ^{1+\beta }\sin \frac{\alpha \pi
}{2}\sin \frac{\beta \pi }{2}}\mathrm{d}\rho \right) \mathrm{d}q \\
&\leq &\frac{\tau ^{2}}{1-\tau }\frac{1}{\sin \frac{\alpha \pi }{2}\sin
\frac{\beta \pi }{2}}\int_{p_{0}}^{\infty }\frac{1}{q^{2-\alpha }}\left(
\int_{1}^{\infty }\frac{1}{\rho ^{1+\beta }}\mathrm{d}\rho \right) \mathrm{d}%
q,\;\;x\in
\mathbb{R}
,\;t>0.
\end{eqnarray*}%
The integral $J_{2}^{-}$ exists and belongs to $C\left(
\mathbb{R}
\right) \cap L^{\infty }\left(
\mathbb{R}
\right) $ with respect to both $x$ and $t.$ Thus, we have that $I^{-},$
given by (\ref{iminus-cstd}), belongs to $C\left(
\mathbb{R}
\right) \cap L^{\infty }\left(
\mathbb{R}
\right) $ with respect to $x,$ and it is a derivative of a continuous
exponentially bounded function with respect to $t.$

Thus, $P$\ has the form%
\begin{equation*}
P\left( x,t\right) =I\left( x,t\right) -\left( \mathrm{i\,}\frac{\partial }{%
\partial t}\left( J_{1}^{+}\left( x,t\right) -J_{1}^{-}\left( x,t\right)
\right) +\frac{\partial ^{2}}{\partial t^{2}}\left( J_{2}^{+}\left(
x,t\right) +J_{2}^{-}\left( x,t\right) \right) \right) \mathrm{e}%
^{s_{0}t},\;\;x\in
\mathbb{R}
,\;t>0,
\end{equation*}%
where $I,$ $J_{1}^{+},$ $J_{1}^{-},$ $J_{2}^{+}$ and $J_{2}^{-}$ are bounded
and continuous functions with respect to $x$ and continuous exponentially
bounded functions with respect to $t.$
\end{proof}

\subsection{Regularization of a generalized solution \label{sec:calc}}

We give a regularization of the generalized solution $u$ to the space-time
fractional Zener wave equation (\ref{eq:fzwe-distr}), which is of particular
importance for the numerical analysis of the problem. We start from the
Fourier and Laplace transform of the solution given by (\ref{u-tilda,het})
and write it as%
\begin{equation}
\widehat{\tilde{u}}(\xi ,s)=\left( \hat{u}_{0}\left( \xi \right) +\frac{1}{s}%
\hat{v}_{0}\left( \xi \right) \right) \widehat{\tilde{K}}\left( \xi
,s\right) ,\;\;\xi \in
\mathbb{R}
,\;\func{Re}s>s_{0},  \label{u-tilda-het}
\end{equation}%
where%
\begin{eqnarray}
\widehat{\tilde{K}}\left( \xi ,s\right) &=&\frac{s}{s^{2}+\frac{1+s^{\alpha }%
}{1+\tau s^{\alpha }}\left\vert \xi \right\vert ^{1+\beta }\sin \frac{\beta
\pi }{2}}=\frac{1}{s}-\widehat{\tilde{Q}}\left( \xi ,s\right) ,\;\;\text{with%
}  \label{K-tilda,het} \\
\widehat{\tilde{Q}}\left( \xi ,s\right) &=&\frac{\frac{1+s^{\alpha }}{1+\tau
s^{\alpha }}\left\vert \xi \right\vert ^{1+\beta }\sin \frac{\beta \pi }{2}}{%
s^{3}+s\frac{1+s^{\alpha }}{1+\tau s^{\alpha }}\left\vert \xi \right\vert
^{1+\beta }\sin \frac{\beta \pi }{2}},\;\;\xi \in
\mathbb{R}
,\;\func{Re}s>s_{0}.  \label{Q-tilda,het}
\end{eqnarray}%
Note that
\begin{equation*}
\widehat{\tilde{K}}\left( \xi ,s\right) =\frac{1}{2\pi }s\widehat{\tilde{P}}%
\left( \xi ,s\right) ,\;\;\text{i.e.,}\;\;K\left( x,t\right) =\frac{1}{2\pi }%
\frac{\partial }{\partial t}P\left( x,t\right) ,\;\;x,\xi \in
\mathbb{R}
,\;t>0,
\end{equation*}%
where $P$ is given by (\ref{P-cstd}). We already know from Theorem \ref%
{th:glavna} that $P$ (and therefore $K$ as well) is a distribution. We
regularize $\widehat{\tilde{K}}$ by multiplying it with the Fourier
transform of the Gaussian%
\begin{equation*}
\delta _{\varepsilon }\left( x\right) =\frac{1}{\varepsilon \sqrt{\pi }}%
\mathrm{e}^{-\frac{x^{2}}{\varepsilon ^{2}}},\;\;x\in
\mathbb{R}
,\;\varepsilon \in \left( 0,1\right] ,
\end{equation*}%
which is a $\delta $-net, i.e., Gaussian in a limiting process $\varepsilon
\rightarrow 0$ represents the Dirac delta distribution. Thus, we have that%
\begin{equation}
\widehat{\tilde{K}}_{\varepsilon }\left( \xi ,s\right) =\widehat{\tilde{K}}%
\left( \xi ,s\right) \mathrm{e}^{-\frac{\left( \varepsilon \xi \right) ^{2}}{%
4}},\;\;\text{where}\;\;\mathcal{F}\left[ \delta _{\varepsilon }\left(
x\right) \right] \left( \xi \right) =\mathrm{e}^{-\frac{\left( \varepsilon
\xi \right) ^{2}}{4}},\;\;\xi \in
\mathbb{R}
,\;\func{Re}s>s_{0},\;\varepsilon \in \left( 0,1\right] ,
\label{K-tilda-het-epsilon}
\end{equation}%
has the inverse Laplace and Fourier transforms which is a function and in a
distributional limit gives the solution kernel $K$ as a distribution.

We summarize these observations in the following theorem, given after we
state the lemma.

\begin{lemma}
\label{lemma_2}Let $\alpha \in \lbrack 0,1)$, $\tau \in (0,1)$ and $\theta
>0 $. Then
\begin{equation*}
\Psi _{\alpha }(s)=s^{2}+\theta \frac{1+s^{\alpha }}{1+\tau s^{\alpha }}%
,\;\;s\in
\mathbb{C}
,
\end{equation*}%
admits exactly two zeros. They are complex-conjugate, located in the left
complex half-plane and each of them is of the multiplicity one.
\end{lemma}

\begin{theorem}
\label{th:poslednja}Let all conditions of Theorem \ref{th:glavna} be
satisfied. Let $u\in \mathcal{K}^{\prime }\left(
\mathbb{R}
^{2}\right) ,$ with support in $%
\mathbb{R}
\times \left[ 0,\infty \right) ,$ be generalized solution to the space-time
fractional Zener wave equation (\ref{eq:fzwe}), with initial (\ref{eq:ic})
and boundary data (\ref{eq:bc}). Then $u$ is of the form:
\begin{equation}
u(x,t)=\left( u_{0}(x)\delta \left( t\right) +v_{0}(x)H\left( t\right)
\right) \ast _{x,t}K\left( x,t\right) ,  \label{eq:sol-u}
\end{equation}%
where $K$ is a distributional limit in $\mathcal{K}^{\prime }\left(
\mathbb{R}
^{2}\right) $:%
\begin{equation}
K\left( x,t\right) =\lim_{\varepsilon \rightarrow 0}K_{\varepsilon }\left(
x,t\right) ,\;\;K_{\varepsilon }\left( x,t\right) =\frac{1}{\pi }%
\int_{0}^{\infty }S\left( \rho ,t\right) \cos \left( \rho x\right) \mathrm{e}%
^{-\frac{\left( \varepsilon \rho \right) ^{2}}{4}}\mathrm{d}\rho ,\;\;x\in
\mathbb{R}
,\;t>0,  \label{Q}
\end{equation}%
with%
\begin{eqnarray}
S\left( \rho ,t\right) &=&\frac{1}{2\pi \mathrm{i}}\int_{0}^{\infty }\left(
\frac{1}{q^{2}+\frac{1+q^{\alpha }\mathrm{e}^{\mathrm{i}\alpha \pi }}{1+\tau
q^{\alpha }\mathrm{e}^{\mathrm{i}\alpha \pi }}\rho ^{1+\beta }\sin \frac{%
\beta \pi }{2}}-\frac{1}{q^{2}+\frac{1+q^{\alpha }\mathrm{e}^{-\mathrm{i}%
\alpha \pi }}{1+\tau q^{\alpha }\mathrm{e}^{-\mathrm{i}\alpha \pi }}\rho
^{1+\beta }\sin \frac{\beta \pi }{2}}\right) q\mathrm{e}^{-qt}\mathrm{d}q
\notag \\
&&+\left. \frac{s\mathrm{e}^{st}}{2s+\frac{\alpha \left( 1-\tau \right)
s^{\alpha -1}}{\left( 1+\tau s^{\alpha }\right) ^{2}}\rho ^{1+\beta }\sin
\frac{\beta \pi }{2}}\right\vert _{s=s_{z}\left( \rho \right) }+\left. \frac{%
s\mathrm{e}^{st}}{2s+\frac{\alpha \left( 1-\tau \right) s^{\alpha -1}}{%
\left( 1+\tau s^{\alpha }\right) ^{2}}\rho ^{1+\beta }\sin \frac{\beta \pi }{%
2}}\right\vert _{s=\bar{s}_{z}\left( \rho \right) }.  \label{inv_lap}
\end{eqnarray}%
and $s_{z}$ are zeros of $\Psi _{\alpha }$ from Lemma \ref{lemma_2}.

In particular, for suitable $s_{0}>0,$ $K_{\varepsilon }\left( x,t\right)
\mathrm{e}^{-s_{0}t}$ is bounded and continuous with respect $x\in
\mathbb{R}
,$ $t>0,$ for every $\varepsilon \in \left( 0,1\right] .$
\end{theorem}

\begin{proof}[Proof of Lemma \protect\ref{lemma_2}]
Let $s=r\mathrm{e}^{\mathrm{i}\varphi }$, $r>0$, $\varphi \in (-\pi ,\pi )$.
We have
\begin{eqnarray}
\func{Re}\Psi _{\alpha }(s) &=&r^{2}\cos (2\varphi )+\theta \frac{1+(1+\tau
)r^{\alpha }\cos (\alpha \varphi )+\tau r^{2\alpha }}{1+2\tau r^{\alpha
}\cos (\alpha \varphi )+\tau ^{2}r^{2\alpha }},  \label{eq:rep} \\
\func{Im}\Psi _{\alpha }(s) &=&r^{2}\sin (2\varphi )+\theta (1-\tau )\frac{%
r^{\alpha }\sin (\alpha \varphi )}{1+2\tau r^{\alpha }\cos (\alpha \varphi
)+\tau ^{2}r^{2\alpha }}.  \label{eq:imp}
\end{eqnarray}

From (\ref{eq:rep}), (\ref{eq:imp}) one can easily see that $\Psi _{\alpha
}(s_{z})=0$ implies $\Psi _{\alpha }(\bar{s}_{z})=0$.

Next we show that if $\func{Re}s_{z}>0$, then such $s_{z}$ cannot be a zero
of $\Psi _{\alpha }$ and therefore zeros must lie in the left complex
half-plane. Suppose $\func{Re}s>0$, i.e., $\varphi \in (-\frac{\pi }{2},%
\frac{\pi }{2})$. Since zeros appears in complex-conjugate pairs, we can
suppose $\varphi \in \lbrack 0,\frac{\pi }{2})$. For $\varphi =0$, we have $%
\Psi _{\alpha }(s)>0$. Since $\alpha \in \lbrack 0,1),$ we have that $\alpha
\varphi ,2\varphi \in (0,\pi )$ and therefore $\sin (2\varphi )>0$ and $\sin
(\alpha \varphi )>0$, which together with $\theta >0$ and $\tau \in (0,1)$
implies $\func{Im}\Psi _{\alpha }(s)>0$, so such $s$ cannot be a zero.

It is left to show that there is only one pair of zeros of $\Psi _{\alpha }$%
. We use the argument principle. Recall, if $f$ is an analytic function
inside and on a regular closed curve $C$, and non-zero on $C$, then number
of zeros of $f$ (counted as many times as its multiplicity) inside the
contour $C$ is equal to the total change in the argument of $f(s)$ as $s$
travels around $C$. For our purpose we choose contour $C=C_{1}\cup C_{2}\cup
C_{3}\cup C_{4}$, parametrized as
\begin{gather*}
C_{1}:s=x\mathrm{e}^{\mathrm{i}\frac{\pi }{2}};\;\;x\in \lbrack
r,R],\;\;\;\;C_{2}:s=R\mathrm{e}^{\mathrm{i}\varphi };\;\;\varphi \in \left[
\frac{\pi }{2},\pi \right] , \\
C_{3}:s=x\mathrm{e}^{\mathrm{i}\pi };\;\;x\in \lbrack r,R],\;\;\;\;C_{4}:s=r%
\mathrm{e}^{\mathrm{i}\varphi };\;\;\varphi \in \left[ \frac{\pi }{2},\pi %
\right] ,
\end{gather*}%
where $r<r_{0},$ $R>R_{0}$ and $r_{0},R_{0}$ are chosen as follows: $r_{0}$
is small enough such that for all $r<r_{0}$ it holds that $\func{Re}\Psi
(s)\sim \theta $ and $\func{Im}\Psi (s)\sim \theta \left( 1-\tau \right)
r^{\alpha }\sin \left( \alpha \pi \right) $ and therefore there are no zeros
for $r<r_{0}$, and $R$ is large enough such that for all $R>R_{0}$ it holds
that $\func{Re}\Psi (s)\sim R^{2}\cos (2\varphi )$ and $\func{Im}\Psi
(s)\sim R^{2}\sin (2\varphi )$.

On the contour $C_{1}$, we have that $\func{Im}\Psi _{\alpha }(s)\geq 0$
(since $\tau \in (0,1)$ and $\alpha \in \lbrack 0,1)$ implies $\sin \frac{%
\alpha \pi }{2}\geq 0$), and $\func{Im}\Psi _{\alpha }(s)\rightarrow 0$ for $%
r,x\rightarrow 0,$ as well as for $R,x\rightarrow \infty $. The real part of
$\Psi _{\alpha }$ varies from $\theta $ (for $r,x\rightarrow 0$) to $-\infty
$ (for $R,x\rightarrow \infty $). Therefore, on $C_{1}$ we have $\Delta \Psi
_{\alpha }(s)=-\pi $.

On the contour $C_{2}$, for $R>R_{0},$ $R_{0}$ large enough, we have%
\begin{equation*}
\func{Im}\Psi \sim R^{2}\sin (2\varphi )+\frac{\theta (1-\tau )\sin (\alpha
\varphi )}{\tau ^{2}}\frac{1}{R^{\alpha }}\sim R^{2}\sin (2\varphi )\leq 0
\end{equation*}%
and we have $\func{Im}\Psi _{\alpha }(s)\rightarrow 0$ for both $\varphi =%
\frac{\pi }{2}$ and $\varphi =\pi $. The real part of $\Psi _{\alpha }$
changes from $-\infty $ (for $\varphi =\frac{\pi }{2}$) to $\infty $ (for $%
\varphi =\pi $) since
\begin{equation*}
\func{Re}\Psi (s)\sim R^{2}\cos (2\varphi )+\frac{\theta }{\tau }\sim
R^{2}\cos (2\varphi ).
\end{equation*}%
So, on $C_{2}$ the change of the argument is $\Delta \Psi _{\alpha }(s)=-\pi
$.

On the contours $C_{3}$ and $C_{4}$ argument does not change. On $C_{3}$
imaginary part of $\Psi _{\alpha }$ is always positive and it tends to zero
for both $R\rightarrow \infty $ and $r\rightarrow 0$, while real part
changes from $\infty $ (for $R\rightarrow \infty $) to $\theta $ (for $%
r\rightarrow 0$), and even if it changes the sign it does not change the
argument of $\Psi _{\alpha }$. On $C_{4}$ it holds that $\Psi _{\alpha
}(s)\sim \theta $ and so there is no changes of argument.

Taking all together we have $\Delta \Psi _{\alpha }(s)=-2\pi $ as $s\in C$,
and by the argument principle there is one zero inside of the contour $C$.
Therefore, there is a unique pair of complex-conjugate numbers in left
complex plain which are zeros of $\Psi _{\alpha }$.
\end{proof}

\begin{proof}[Proof of Theorem \protect\ref{th:poslednja}]
Let
\begin{equation}
\widehat{\tilde{K}}_{\varepsilon }\left( \xi ,s\right) =\frac{s}{s^{2}+\frac{%
1+s^{\alpha }}{1+\tau s^{\alpha }}\left\vert \xi \right\vert ^{1+\beta }\sin
\frac{\beta \pi }{2}}\mathrm{e}^{-\frac{\left( \varepsilon \xi \right) ^{2}}{%
4}}=\frac{1}{s}\mathrm{e}^{-\frac{\left( \varepsilon \xi \right) ^{2}}{4}}-%
\widehat{\tilde{Q}}_{\varepsilon }\left( \xi ,s\right) ,\;\;\xi \in
\mathbb{R}
,\;\func{Re}s>s_{0},  \label{K-tilda-het-epsilon-1}
\end{equation}%
by (\ref{K-tilda,het}) and (\ref{K-tilda-het-epsilon}). We shall prove that%
\begin{equation*}
\widehat{\tilde{Q}}_{\varepsilon }\left( \xi ,s\right) =\widehat{\tilde{Q}}%
\left( \xi ,s\right) \mathrm{e}^{-\frac{\left( \varepsilon \xi \right) ^{2}}{%
4}}=\frac{\frac{1+s^{\alpha }}{1+\tau s^{\alpha }}\left\vert \xi \right\vert
^{1+\beta }\sin \frac{\beta \pi }{2}}{s^{3}+s\frac{1+s^{\alpha }}{1+\tau
s^{\alpha }}\left\vert \xi \right\vert ^{1+\beta }\sin \frac{\beta \pi }{2}}%
\mathrm{e}^{-\frac{\left( \varepsilon \xi \right) ^{2}}{4}},\;\;\xi \in
\mathbb{R}
,\;\func{Re}s>s_{0},
\end{equation*}%
see (\ref{Q-tilda,het}), has the inverse Laplace and Fourier transforms by
examining the convergence of the double integral ($x\in
\mathbb{R}
,$ $t>0,$ $\varepsilon \in \left( 0,1\right] $)%
\begin{eqnarray}
Q_{\varepsilon }\left( x,t\right) &=&\frac{1}{\left( 2\pi \right) ^{2}%
\mathrm{i}}\int_{s_{0}-\mathrm{i}\infty }^{s_{0}+\mathrm{i}\infty }\left(
\int_{-\infty }^{\infty }\widehat{\tilde{Q}}_{\varepsilon }\left( \xi
,s\right) \mathrm{e}^{\mathrm{i}\xi x}\mathrm{d}\xi \right) \mathrm{e}^{st}%
\mathrm{d}s  \notag \\
&=&\frac{1}{2\pi ^{2}}\left( J_{\varepsilon }\left( x,t\right)
+J_{\varepsilon }^{+}\left( x,t\right) +J_{\varepsilon }^{-}\left(
x,t\right) \right) \mathrm{e}^{s_{0}t},  \label{Q-epsilon}
\end{eqnarray}%
with ($x\in
\mathbb{R}
,$ $t>0,$ $\varepsilon \in \left( 0,1\right] $)%
\begin{eqnarray}
J_{\varepsilon }\left( x,t\right) &=&\int_{-p_{0}}^{p_{0}}\int_{0}^{\infty
}\left. \frac{\frac{1+s^{\alpha }}{1+\tau s^{\alpha }}\rho ^{1+\beta }\sin
\frac{\beta \pi }{2}}{s^{3}+s\frac{1+s^{\alpha }}{1+\tau s^{\alpha }}\rho
^{1+\beta }\sin \frac{\beta \pi }{2}}\right\vert _{s=s_{0}+\mathrm{i}p}\cos
\left( \rho x\right) \mathrm{e}^{-\frac{\left( \varepsilon \rho \right) ^{2}%
}{4}}\mathrm{e}^{\mathrm{i}pt}\mathrm{d}\rho \,\mathrm{d}p,  \label{jot} \\
J_{\varepsilon }^{+}\left( x,t\right) &=&\int_{p_{0}}^{\infty
}\int_{0}^{\infty }\left. \frac{\frac{1+s^{\alpha }}{1+\tau s^{\alpha }}\rho
^{1+\beta }\sin \frac{\beta \pi }{2}}{s^{3}+s\frac{1+s^{\alpha }}{1+\tau
s^{\alpha }}\rho ^{1+\beta }\sin \frac{\beta \pi }{2}}\right\vert _{s=s_{0}+%
\mathrm{i}p}\cos \left( \rho x\right) \mathrm{e}^{-\frac{\left( \varepsilon
\rho \right) ^{2}}{4}}\mathrm{e}^{\mathrm{i}pt}\mathrm{d}\rho \,\mathrm{d}p,
\label{jot+} \\
J_{\varepsilon }^{-}\left( x,t\right) &=&\int_{-\infty
}^{p_{0}}\int_{0}^{\infty }\left. \frac{\frac{1+s^{\alpha }}{1+\tau
s^{\alpha }}\rho ^{1+\beta }\sin \frac{\beta \pi }{2}}{s^{3}+s\frac{%
1+s^{\alpha }}{1+\tau s^{\alpha }}\rho ^{1+\beta }\sin \frac{\beta \pi }{2}}%
\right\vert _{s=s_{0}+\mathrm{i}p}\cos \left( \rho x\right) \mathrm{e}^{-%
\frac{\left( \varepsilon \rho \right) ^{2}}{4}}\mathrm{e}^{\mathrm{i}pt}%
\mathrm{d}\rho \,\mathrm{d}p,  \label{jot-}
\end{eqnarray}%
where we introduced the parametrization $s=s_{0}+\mathrm{i}p,$ $p\in \left(
-\infty ,\infty \right) $ in (\ref{Q-epsilon}) and used the fact that $%
\widehat{\tilde{Q}}_{\varepsilon }$ is an even function in $\xi .$ From (\ref%
{jot}), we have ($x\in
\mathbb{R}
,$ $t>0,$ $\varepsilon \in \left( 0,1\right] $)%
\begin{equation*}
\left\vert J_{\varepsilon }\left( x,t\right) \right\vert \leq
\int_{-p_{0}}^{p_{0}}\int_{0}^{\infty }\frac{\left\vert \left. \frac{%
1+s^{\alpha }}{1+\tau s^{\alpha }}\right\vert _{s=s_{0}+\mathrm{i}%
p}\right\vert \rho ^{1+\beta }}{\left\vert \left[ s^{3}+s\frac{1+s^{\alpha }%
}{1+\tau s^{\alpha }}\rho ^{1+\beta }\sin \frac{\beta \pi }{2}\right]
_{s=s_{0}+\mathrm{i}p}\right\vert }\mathrm{e}^{-\frac{\left( \varepsilon
\rho \right) ^{2}}{4}}\mathrm{d}\rho \,\mathrm{d}p<\infty .
\end{equation*}%
Let us estimate the integral given by (\ref{jot+}) as ($x\in
\mathbb{R}
,$ $t>0,$ $\varepsilon \in \left( 0,1\right] $)%
\begin{eqnarray}
\left\vert J_{\varepsilon }^{+}\left( x,t\right) \right\vert &\leq
&\int_{p_{0}}^{\infty }\int_{0}^{\infty }\frac{\left\vert \left. \frac{%
1+s^{\alpha }}{1+\tau s^{\alpha }}\right\vert _{s=s_{0}+\mathrm{i}%
p}\right\vert \rho ^{1+\beta }}{\left\vert \left[ s^{3}+s\frac{1+s^{\alpha }%
}{1+\tau s^{\alpha }}\rho ^{1+\beta }\sin \frac{\beta \pi }{2}\right]
_{s=s_{0}+\mathrm{i}p}\right\vert }\mathrm{e}^{-\frac{\left( \varepsilon
\rho \right) ^{2}}{4}}\mathrm{d}\rho \,\mathrm{d}p  \notag \\
&\leq &\int_{p_{0}}^{\infty }\int_{0}^{\infty }\frac{\left\vert \left. \frac{%
1+s^{\alpha }}{1+\tau s^{\alpha }}\right\vert _{s=s_{0}+\mathrm{i}%
p}\right\vert \rho ^{1+\beta }}{\left\vert \func{Im}\left( \left[ s^{3}+s%
\frac{1+s^{\alpha }}{1+\tau s^{\alpha }}\rho ^{1+\beta }\sin \frac{\beta \pi
}{2}\right] _{s=s_{0}+\mathrm{i}p}\right) \right\vert }\mathrm{e}^{-\frac{%
\left( \varepsilon \rho \right) ^{2}}{4}}\mathrm{d}\rho \,\mathrm{d}p.
\label{jot+1}
\end{eqnarray}%
We have $\func{Im}\left( s_{0}+\mathrm{i}p\right) ^{3}=-p^{3}+3s_{0}^{2}p,$ $%
\left\vert \left. \frac{1+s^{\alpha }}{1+\tau s^{\alpha }}\right\vert
_{s=s_{0}+\mathrm{i}p}\right\vert \sim \frac{1}{\tau }$ and
\begin{equation*}
\func{Im}\left( \left[ s\frac{1+s^{\alpha }}{1+\tau s^{\alpha }}\rho
^{1+\beta }\sin \frac{\beta \pi }{2}\right] _{s=s_{0}+\mathrm{i}p}\right)
\sim \frac{1}{\tau }p+s_{0}\frac{1-\tau }{\tau ^{2}}\frac{1}{p^{\alpha }}%
\sin \frac{\alpha \pi }{2},\;\;\text{as}\;\;p\rightarrow \infty ,
\end{equation*}%
since%
\begin{eqnarray*}
\func{Re}\left( \frac{1+s^{\alpha }}{1+\tau s^{\alpha }}\right) &=&\frac{%
1+(1+\tau )r^{\alpha }\cos (\alpha \varphi )+\tau r^{2\alpha }}{1+2\tau
r^{\alpha }\cos (\alpha \varphi )+\tau ^{2}r^{2\alpha }}\sim \frac{1}{\tau }%
,\;\;\text{as}\;\;r\rightarrow \infty , \\
\func{Im}\left( \frac{1+s^{\alpha }}{1+\tau s^{\alpha }}\right) &=&(1-\tau )%
\frac{r^{\alpha }\sin (\alpha \varphi )}{1+2\tau r^{\alpha }\cos (\alpha
\varphi )+\tau ^{2}r^{2\alpha }}\sim \frac{1-\tau }{\tau ^{2}}\frac{1}{%
r^{\alpha }}\sin (\alpha \varphi ),\;\;\text{as}\;\;r\rightarrow \infty ,
\end{eqnarray*}%
and thus for $r=\sqrt{s_{0}^{2}+p^{2}},$ $\tan \varphi =\frac{s_{0}}{p}$
\begin{eqnarray}
&&\func{Im}\left( \left[ s^{3}+s\frac{1+s^{\alpha }}{1+\tau s^{\alpha }}\rho
^{1+\beta }\sin \frac{\beta \pi }{2}\right] _{s=s_{0}+\mathrm{i}p}\right)
\notag \\
&&\qquad \qquad \sim -p^{3}+3s_{0}^{2}p+\frac{1}{\tau }p+s_{0}\frac{1-\tau }{%
\tau ^{2}}\frac{1}{p^{\alpha }}\rho ^{1+\beta }\sin \frac{\alpha \pi }{2}%
\sin \frac{\beta \pi }{2},\;\;\text{as}\;\;p\rightarrow \infty .
\label{im-p-na3}
\end{eqnarray}%
We choose $p_{0}$ so that (\ref{im-p-na3}) becomes%
\begin{equation*}
\func{Im}\left( \left[ s^{3}+s\frac{1+s^{\alpha }}{1+\tau s^{\alpha }}\rho
^{1+\beta }\sin \frac{\beta \pi }{2}\right] _{s=s_{0}+\mathrm{i}p}\right)
\sim -p^{3},\;\;\text{as}\;\;p\rightarrow \infty ,
\end{equation*}%
Thus, for (\ref{jot+1}) we have%
\begin{equation*}
\left\vert J_{\varepsilon }^{+}\left( x,t\right) \right\vert \leq
\int_{p_{0}}^{\infty }\int_{0}^{\infty }\frac{\rho ^{1+\beta }}{p^{3}}%
\mathrm{e}^{-\frac{\left( \varepsilon \rho \right) ^{2}}{4}}\mathrm{d}\rho \,%
\mathrm{d}p<\infty .
\end{equation*}%
Using the same arguments as for (\ref{jot+}), we can prove that $%
J_{\varepsilon }^{-},$ given by (\ref{jot-}), is also absolutely integrable.

We proved that $Q_{\varepsilon },$ given by (\ref{Q-epsilon}), has the
inverse Laplace and Fourier transforms and therefore, by (\ref%
{K-tilda-het-epsilon-1}), we have%
\begin{eqnarray*}
\widehat{\tilde{K}}_{\varepsilon }\left( \xi ,s\right) &=&\frac{1}{s}\mathrm{%
e}^{-\frac{\left( \varepsilon \xi \right) ^{2}}{4}}-\widehat{\tilde{Q}}%
_{\varepsilon }\left( \xi ,s\right) ,\;\;\xi \in
\mathbb{R}
,\;\func{Re}s>s_{0},\;\varepsilon \in \left( 0,1\right] ,\;\;\text{i.e.,} \\
K_{\varepsilon }\left( x,t\right) &=&H\left( t\right) \delta _{\varepsilon
}\left( x\right) -Q_{\varepsilon }\left( x,t\right) ,\;\;x\in
\mathbb{R}
,\;t>0.
\end{eqnarray*}%
Thus, in (\ref{K-tilda-het-epsilon-1}) we can first invert the Laplace
transform and subsequently the Fourier transform. The Fourier transform of
the solution kernel $\hat{K}_{\varepsilon }$ is obtained by the use the
inversion formula of the Laplace transform%
\begin{equation}
\hat{K}_{\varepsilon }(\rho ,t)=\frac{1}{2\pi \mathrm{i}}\int_{s_{0}-\mathrm{%
i}\infty }^{s_{0}+\mathrm{i}\infty }\widehat{\tilde{K}}_{\varepsilon }(\rho
,s)\mathrm{e}^{st}\mathrm{d}s,\;\;a\geq 0,  \label{inv_laplase}
\end{equation}%
where $\widehat{\tilde{K}}_{\varepsilon }$ is given by (\ref%
{K-tilda-het-epsilon-1}) and the complex integration along the contour $%
\Gamma =\Gamma _{1}\cup \Gamma _{2}\cup \Gamma _{r}\cup \Gamma _{3}\cup
\Gamma _{4}\cup \gamma _{0}$, presented in Figure \ref{kontura}.
\begin{figure}[tbh]
\centering
\includegraphics[scale=0.6]{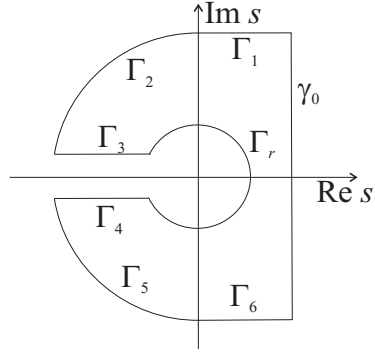}
\caption{Integration contour $\Gamma $.}
\label{kontura}
\end{figure}
The contour $\Gamma $ is parameterized by:
\begin{gather*}
\Gamma _{1}:s=R\mathrm{e}^{\mathrm{i}\varphi },\;\;\varphi _{0}<\varphi <\pi
;\;\;\;\;\Gamma _{2}:s=q\mathrm{e}^{\mathrm{i}\pi },\;\;-R<-q<-r; \\
\Gamma _{r}:r\mathrm{e}^{\mathrm{i}\varphi },\;\;-\pi <-\varphi <\pi
;\;\;\;\;\Gamma _{3}:s=q\mathrm{e}^{-\mathrm{i}\pi },\;\;r<q<R; \\
\Gamma _{4}:s=R\mathrm{e}^{\mathrm{i}\varphi },\;\;-\pi <\varphi <\varphi
_{0};\;\;\;\;\gamma _{0}:s=s_{0}(1+\mathrm{i}\tan \varphi ),\;\;-\varphi
_{0}<\varphi <\varphi _{0},
\end{gather*}%
for arbitrary chosen $R>0$ and $0<r<R$, and $\varphi _{0}=\arccos \frac{s_{0}%
}{R}$. By the Cauchy residues theorem and results of Lemma \ref{lemma_2} we
obtain:%
\begin{equation}
\frac{1}{2\pi \mathrm{i}}\oint_{\Gamma }\widehat{\tilde{K}}_{\varepsilon
}(\rho ,s)\mathrm{e}^{st}\mathrm{d}s=\func{Res}\left( \widehat{\tilde{K}}%
_{\varepsilon }(\rho ,s)\mathrm{e}^{st},s_{z}\left( \rho \right) \right) +%
\func{Res}\left( \widehat{\tilde{K}}_{\varepsilon }(\rho ,s)\mathrm{e}^{st},%
\bar{s}_{z}\left( \rho \right) \right) .  \label{caushy_res}
\end{equation}%
Now, one shows (see, for example, \cite{KOZ10} for similar calculations)
that in (\ref{caushy_res}), when $R$ tends to infinity and $r$ tends to
zero, integrals along contours $\Gamma _{1}$, $\Gamma _{4}$ and $\Gamma _{r}$
tend to zero. The integrals along contours $\Gamma _{2}$ and $\Gamma _{3}$
in limiting process (when $R$ tends to infinity and $r$ tends to zero) read (%
$\rho \geq 0,$ $t>0$)%
\begin{eqnarray*}
\lim_{R\rightarrow \infty ,r\rightarrow 0}\int_{\Gamma _{2}}\widehat{\tilde{K%
}}_{\varepsilon }(\rho ,s)\mathrm{e}^{st}\mathrm{d}s &=&-\mathrm{e}^{-\frac{%
\left( \varepsilon \xi \right) ^{2}}{4}}\int_{0}^{\infty }\frac{q}{q^{2}+%
\frac{1+q^{\alpha }\mathrm{e}^{\mathrm{i}\alpha \pi }}{1+\tau q^{\alpha }%
\mathrm{e}^{\mathrm{i}\alpha \pi }}\rho ^{1+\beta }\sin \frac{\beta \pi }{2}}%
\mathrm{e}^{-qt}\mathrm{d}q, \\
\lim_{R\rightarrow \infty ,r\rightarrow 0}\int_{\Gamma _{3}}\widehat{\tilde{K%
}}_{\varepsilon }(\rho ,s)\mathrm{e}^{st}\mathrm{d}s &=&\mathrm{e}^{-\frac{%
\left( \varepsilon \xi \right) ^{2}}{4}}\int_{0}^{\infty }\frac{q}{q^{2}+%
\frac{1+q^{\alpha }\mathrm{e}^{-\mathrm{i}\alpha \pi }}{1+\tau q^{\alpha }%
\mathrm{e}^{-\mathrm{i}\alpha \pi }}\rho ^{1+\beta }\sin \frac{\beta \pi }{2}%
}\mathrm{e}^{-qt}\mathrm{d}q.
\end{eqnarray*}%
By Lemma \ref{lemma_2}, we have that the residues in (\ref{caushy_res}) read
($\rho \geq 0,$ $t>0$)%
\begin{eqnarray*}
\func{Res}\left( \widehat{\tilde{K}}_{\varepsilon }(\rho ,s)\mathrm{e}%
^{st},s_{z}\left( \rho \right) \right) &=&\left. \frac{s\mathrm{e}^{st}}{2s+%
\frac{\alpha \left( 1-\tau \right) s^{\alpha -1}}{\left( 1+\tau s^{\alpha
}\right) ^{2}}\rho ^{1+\beta }\sin \frac{\beta \pi }{2}}\right\vert
_{s=s_{z}\left( \rho \right) }\mathrm{e}^{-\frac{\left( \varepsilon \xi
\right) ^{2}}{4}}, \\
\func{Res}\left( \widehat{\tilde{K}}_{\varepsilon }(\rho ,s)\mathrm{e}^{st},%
\bar{s}_{z}\left( \rho \right) \right) &=&\left. \frac{s\mathrm{e}^{st}}{2s+%
\frac{\alpha \left( 1-\tau \right) s^{\alpha -1}}{\left( 1+\tau s^{\alpha
}\right) ^{2}}\rho ^{1+\beta }\sin \frac{\beta \pi }{2}}\right\vert _{s=\bar{%
s}_{z}\left( \rho \right) }\mathrm{e}^{-\frac{\left( \varepsilon \xi \right)
^{2}}{4}}.
\end{eqnarray*}%
Integral along the contour $\gamma _{0}$ in limiting process tends to the
integral on the right-hand side of (\ref{inv_laplase}) and therefore,
putting all together in (\ref{caushy_res}) we obtain%
\begin{equation*}
\hat{K}_{\varepsilon }\left( \rho ,t\right) =S\left( \rho ,t\right) \mathrm{e%
}^{-\frac{\left( \varepsilon \rho \right) ^{2}}{4}},\;\;x\in
\mathbb{R}
,\;t>0,
\end{equation*}%
with $S$ given by (\ref{inv_lap}). The inverse Fourier transform of such
obtained $\hat{K}_{\varepsilon }$ reads%
\begin{equation*}
K_{\varepsilon }\left( x,t\right) =\frac{1}{\pi }\int_{0}^{\infty }S\left(
\rho ,t\right) \cos \left( \rho x\right) \mathrm{e}^{-\frac{\left(
\varepsilon \rho \right) ^{2}}{4}}\mathrm{d}\rho ,\;\;x\in
\mathbb{R}
,\;t>0,
\end{equation*}%
which in the distributional limit when $\varepsilon \rightarrow 0$ gives the
solution kernel $K$ in the form (\ref{Q}).
\end{proof}

\section{Dependence of a solution on parameters $\protect\alpha $ and $%
\protect\beta $ \label{cases}}

We examine the solutions in the limiting cases of system (\ref{eq:system1})
- (\ref{eq:system3}), or equivalently (\ref{eq:fzwe}), subject to (\ref%
{eq:ic}), (\ref{eq:bc}) in the view of Remark \ref{cases-rem}. In all cases
we write the solution $u$ of Theorem \ref{th:poslednja} as
\begin{equation*}
u\left( x,t\right) =\left( u_{0}\left( x\right) \delta \left( t\right)
+v_{0}\left( x\right) H\left( t\right) \right) \ast _{x,t}K_{\alpha ,\beta
}\left( x,t\right) ,\;\;x\in
\mathbb{R}
,\;t>0,
\end{equation*}%
where the inverse Fourier transform of $\widehat{\tilde{K}}_{\alpha ,\beta
}, $ (\ref{K-tilda,het}), is given by%
\begin{equation}
\tilde{K}_{\alpha ,\beta }\left( x,s\right) =\frac{1}{\pi }\int_{0}^{\infty }%
\frac{s}{s^{2}+\frac{1+s^{\alpha }}{1+\tau s^{\alpha }}\rho ^{1+\beta }\sin
\frac{\beta \pi }{2}}\cos (\rho x)\mathrm{d}\rho ,\;\;x\in
\mathbb{R}
,\;\func{Re}s>0.  \label{K-tilda}
\end{equation}%
Note that the integral in (\ref{K-tilda}) is written formally and it denotes
the inverse Fourier transform. As it will be seen, it may either converge,
or diverge representing a distribution.

We are interested in the behavior of $K_{\alpha ,\beta }$ for $\alpha $ and $%
\beta $ tending to zero and one. We expect that the solution kernel $%
K_{\alpha ,\beta }$ tends to solution kernels in specific cases. From the
form of $K_{\alpha ,\beta },$ given by (\ref{Q}), this cannot be easily
seen. However, numerical examples, see Section \ref{sec:ex}, suggest that
this holds true. What can be seen analytically is that the Laplace transform
of $K_{\alpha ,\beta }$ tends to the the Laplace transforms of solution
kernels in specific cases.

When $\beta \rightarrow 0,$ then, in the sense of distributions,%
\begin{equation*}
\tilde{K}_{\alpha ,0}\left( x,s\right) =\frac{1}{\pi }\int_{0}^{\infty }%
\frac{1}{s}\cos \left( \rho x\right) \mathrm{d}\rho ,\;\;x\in
\mathbb{R}
,\;\func{Re}s>0,
\end{equation*}%
and the solution kernel $K_{\alpha ,0}$ is of the form%
\begin{equation}
K_{\alpha ,0}\left( x,t\right) =\delta \left( x\right) H\left( t\right)
,\;\;x\in
\mathbb{R}
,\;t>0.  \label{delta-h}
\end{equation}%
This is the case of the non-propagating disturbance, if the initial velocity
is zero and the solution is given by (\ref{non-prop}) with $v_{0}=0$.
Therefore, regardless of the parameter $\alpha ,$ when $\beta $ tends to
zero, solution kernel tends to (\ref{delta-h}). This supports the idea from
Remark \ref{cases-rem}, $\left( ii\right) ,$ that our system can be useful
in modelling materials which resist the propagation of the initial
disturbance.

When $\beta \rightarrow 1,$ we obtain the case of the time-fractional Zener
wave equation, studied in \cite{KOZ10}. In this case%
\begin{equation*}
\tilde{K}_{\alpha ,1}\left( x,s\right) =\frac{1}{\pi }\int_{0}^{\infty }%
\frac{s}{s^{2}+\frac{1+s^{\alpha }}{1+\tau s^{\alpha }}\rho ^{2}}\cos \left(
\rho x\right) \mathrm{d}\rho ,\;\;x\in
\mathbb{R}
,\;\func{Re}s>0,
\end{equation*}%
and the calculation similar to one presented in \cite{KOZ10}, leads to
\begin{equation}
K_{\alpha ,1}\left( x,t\right) =\frac{1}{4\pi \mathrm{i}}\int_{0}^{\infty
}\left( f_{-}\left( q\right) \mathrm{e}^{\left\vert x\right\vert
qf_{-}\left( q\right) }-f_{+}\left( q\right) \mathrm{e}^{\left\vert
x\right\vert qf_{+}\left( q\right) }\right) \mathrm{e}^{-qt}\mathrm{d}%
q,\;\;x\in
\mathbb{R}
,\;t>0,  \label{koz}
\end{equation}%
with%
\begin{equation*}
f_{+}\left( q\right) =\sqrt{\frac{1+\tau q^{\alpha }\mathrm{e}^{\mathrm{i}%
\alpha \pi }}{1+q^{\alpha }\mathrm{e}^{\mathrm{i}\alpha \pi }}},\;\;\text{and%
}\;\;f_{-}\left( q\right) =\sqrt{\frac{1+\tau q^{\alpha }\mathrm{e}^{-%
\mathrm{i}\alpha \pi }}{1+q^{\alpha }\mathrm{e}^{-\mathrm{i}\alpha \pi }}}%
,\;\;q>0.
\end{equation*}%
We note that in \cite{KOZ10} the solution is given in a slightly different
form.

When $\alpha \rightarrow 0,$ then%
\begin{equation*}
\tilde{K}_{0,\beta }\left( x,s\right) =\frac{1}{\pi }\int_{0}^{\infty }\frac{%
s}{s^{2}+\frac{2}{1+\tau }\rho ^{1+\beta }\sin \frac{\beta \pi }{2}}\cos
\left( \rho x\right) \mathrm{d}\rho ,\;\;x\in
\mathbb{R}
,\;\func{Re}s>0.
\end{equation*}%
Using $\mathcal{L}^{-1}\left[ \frac{s}{s^{2}+\omega ^{2}}\right] \left(
t\right) =\cos \left( \omega t\right) $ one easily comes to%
\begin{equation*}
K_{0,\beta }\left( x,t\right) =\frac{1}{\pi }\int_{0}^{\infty }\cos \left( t%
\sqrt{\frac{2}{1+\tau }\rho ^{1+\beta }\sin \frac{\beta \pi }{2}}\right)
\cos \left( \rho x\right) \mathrm{d}\rho ,\;\;x\in
\mathbb{R}
,\;t>0,
\end{equation*}%
in the sense of distributions, which can be transformed to
\begin{eqnarray}
K_{0,\beta }\left( x,t\right) &=&\frac{1}{2\pi }\int_{0}^{\infty }\left(
\cos \left( \left( x+ct\sqrt{\frac{1}{\rho ^{1-\beta }}\sin \frac{\beta \pi
}{2}}\right) \rho \right) \right.  \notag \\
&&+\left. \cos \left( \left( x-ct\sqrt{\frac{1}{\rho ^{1-\beta }}\sin \frac{%
\beta \pi }{2}}\right) \rho \right) \right) \mathrm{d}\rho ,\;\;x\in
\mathbb{R}
,\;t>0,  \label{alfa0;beta0,1}
\end{eqnarray}%
where $c=\sqrt{\frac{2}{1+\tau }}.$ This is the case of the space-fractional
wave equation studied in \cite{A-S-09}. Note that the solution in \cite%
{A-S-09} is given in a different form. In this case, from (\ref%
{alfa0;beta0,1}), one can recover solution kernels for $\beta =0$ and $\beta
=1.$

If we put $\beta =0$ in (\ref{alfa0;beta0,1}), we obtain%
\begin{equation*}
K_{0,0}\left( x,t\right) =\frac{1}{\pi }\int_{0}^{\infty }\cos \left( x\rho
\right) \mathrm{d}\rho =\delta \left( x\right) ,\;\;x\in
\mathbb{R}
,\;t>0,
\end{equation*}%
i.e., (\ref{delta-h}).

For $\beta =1$ in (\ref{alfa0;beta0,1}), we obtain, in the sense of
distributions,
\begin{eqnarray*}
K_{0,1}\left( x,t\right) &=&\frac{1}{2\pi }\int_{0}^{\infty }\left( \cos
\left( \left( x+ct\right) \rho \right) +\cos \left( \left( x-ct\right) \rho
\right) \right) \mathrm{d}\rho ,\;\;x\in
\mathbb{R}
,\;t>0, \\
&=&\frac{1}{2}\left( \delta \left( x+ct\right) +\delta \left( x-ct\right)
\right) ,
\end{eqnarray*}%
with $c=\sqrt{\frac{2}{1+\tau }}.$ This is the solution kernel for the
classical wave equation.

\begin{remark}[Question of the wave speed.]
Note that in all cases when $\beta =1$, one obtains the finite and constant
speed of wave propagation: $c=\sqrt{\frac{2}{1+\tau }}$ in the case $\alpha
=0$ and $c=\frac{1}{\sqrt{\tau }}$ in the case $\alpha \in \left( 0,1\right)
.$ For $\beta =0$ the wave speed is zero. The case $\beta \in \left(
0,1\right) $ is much more complicated for investigation. It seems that one
would need to employ other technics, e.g. the theory of Fourier integral
operators in order to reach some conclusions. However, we tend to believe
that in such cases wave speed is not constant, depends on spatial variable $%
x $ and parameter $\beta ,$ i.e., $c=c\left( x,\beta \right) .$
\end{remark}

\section{Numerical examples \label{sec:ex}}

We examine the qualitative properties of the solution to space-time
fractional Zener wave equation (\ref{eq:fzwe}). Further, we investigate the
influence of the orders $\alpha $ and $\beta $ of the, respective, time and
space fractionalization of the constitutive equation and strain measure.
Also, we numerically compare solution to (\ref{eq:fzwe}) with the solutions
to time fractional Zener wave equation (\ref{tfzwe}), that represents the
limiting case $\beta =1$ in (\ref{eq:fzwe}). Both equations are subject to
initial conditions $u_{0}=\delta ,$ $v_{0}=0.$ In this case, the solution to
(\ref{eq:fzwe}), given by (\ref{eq:sol-u}), (\ref{Q}), becomes%
\begin{equation}
u\left( x,t\right) =\delta \left( x\right) \ast _{x}K\left( x,t\right)
=K\left( x,t\right) ,\;\;x\in
\mathbb{R}
,\;t>0.  \label{u-za-delta}
\end{equation}%
Since $K,$ and therefore $u$ is a distribution in $x,$ it cannot be plotted.
Thus, we use the regularization $K_{\varepsilon }$ of the solution kernel so
that (\ref{u-za-delta}) becomes
\begin{equation}
u_{\varepsilon }\left( x,t\right) =\frac{1}{\pi }\int_{0}^{\infty }\hat{K}%
(\rho ,t)\mathrm{e}^{-\frac{\left( \varepsilon \rho \right) ^{2}}{4}}\cos
(\rho x)\mathrm{d}\rho ,\;\;x\in
\mathbb{R}
,\;t>0.  \label{u-za-delta-gaus}
\end{equation}

In all figures that are to follow, we present the displacement field only on
the half-axis $x\geq 0,$ since the field is symmetric with respect to
displacement axis. Figure \ref{fig.1} presents the plot of the displacement
versus coordinate, obtained according to (\ref{u-za-delta-gaus}) for several
time instants, while the other parameters of the model are: $\alpha =0.25,$ $%
\beta =0.45,$ $\tau =0.1,$ $\varepsilon =0.01.$
\begin{figure}[tbph]
\centering
\includegraphics[scale=1]{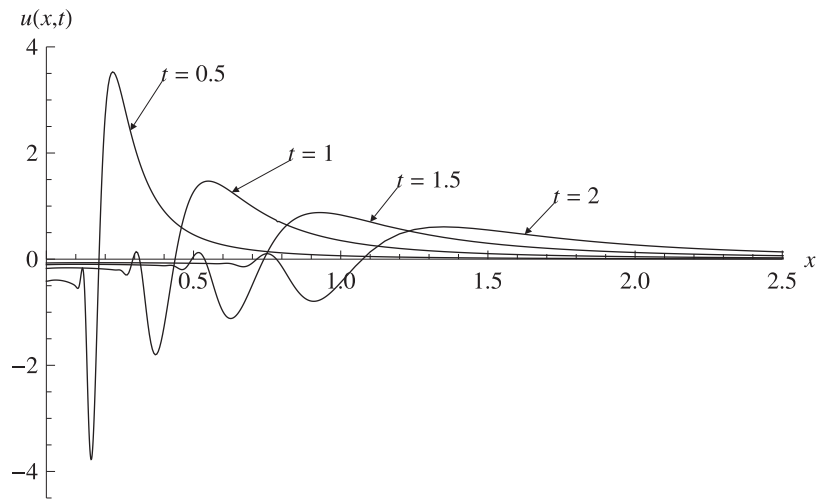}
\caption{Displacement $u\left( x,t\right) $ at $t\in \left\{
0.5,1,1.5,2\right\} $ as a function of $x$ as a solution of (\protect\ref%
{eq:fzwe}).}
\label{fig.1}
\end{figure}
From Figure \ref{fig.1} we see that as time increases, the height of the
peaks decreases, since the energy introduced by the initial disturbance
field is being dissipated. This is the consequence of the viscoelastic
properties of the material.

In Figure \ref{fig.2} we compared the displacements obtained as a solutions
for non-local (\ref{eq:fzwe}) and local (\ref{tfzwe}) wave equations, given
by (\ref{u-za-delta-gaus}) and (\ref{koz}), respectively.
\begin{figure}[tbph]
\centering
\includegraphics[scale=1]{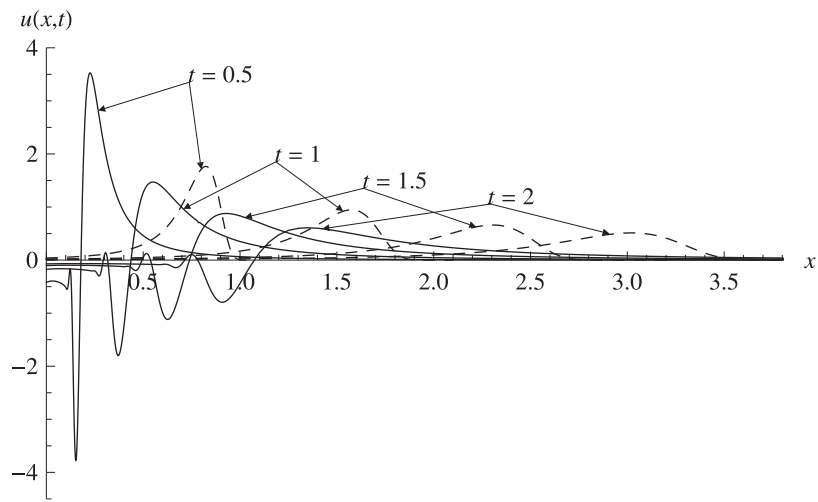}
\caption{Displacement $u\left( x,t\right) $ at $t\in \left\{
0.5,1,1.5,2\right\} $ as a function of $x$ as a solution of: (\protect\ref%
{eq:fzwe}) - solid line and (\protect\ref{tfzwe}) - dashed line.}
\label{fig.2}
\end{figure}
Apart from $\beta =0.45$ in (\ref{eq:fzwe}) and $\beta =1$ in (\ref{tfzwe})
other parameters in both models are as above. The effect of non-locality
introduced in the strain measure is observed, since at fixed time-instant,
apart from the primary peak that exists in both models, in the non-local one
there are secondary peaks in the displacement field. These secondary peaks
reflect the influence of the oscillations of a certain material point to
other material points in a medium. Thus, when the initial disturbance
propagates (this is reflected by the existence of the primary peak), due to
the non-locality, the secondary peaks reflect the residual influence of the
disturbance transported by the primary peak. Not only that non-locality
changes the number of peaks at a certain time-instant, but also the shape of
the primary peak is changed. The primary peak in the non-local model
(compared to the local one) is higher and placed closer to the origin - the
point where the initial Dirac-type disturbance field is introduced.

The aim of the following figures is to show the influence of changing the
non-locality parameter $\beta .$ All other parameters are as above. Figure %
\ref{fig.3} presents plots of displacements for various values of $\beta .$
\begin{figure}[tbph]
\centering
\includegraphics[scale=1]{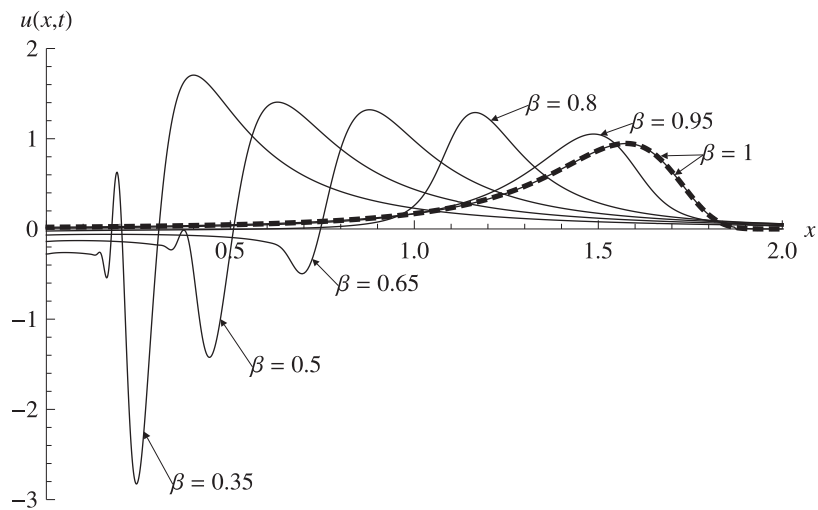}
\caption{Displacement $u\left( x,t\right) $ at $t=1$ as a function of $x$ as
a solution of: (\protect\ref{eq:fzwe}) - solid line and (\protect\ref{tfzwe}%
) - dashed line.}
\label{fig.3}
\end{figure}
From Figure \ref{fig.3} one sees that as the non-locality parameter $\beta $
increases, the effects of non-locality decrease, since the height of the
secondary peaks decreases and eventually the secondary peaks cease to exist.
Also, as $\beta $ increases, the height of the primary peaks decreases and
its position increases, being further from the point of the initial
Dirac-type disturbance. In the limiting case $\beta =1$ the displacement
curve of non-local model (\ref{eq:fzwe}) overlaps with the displacement
curve of the local model (\ref{tfzwe}). Figures \ref{fig.4} and \ref{fig.5}
present the displacement field for smaller values of $\beta .$
\begin{figure}[h]
\begin{minipage}{65mm}
 \centering
 \includegraphics[scale=0.7]{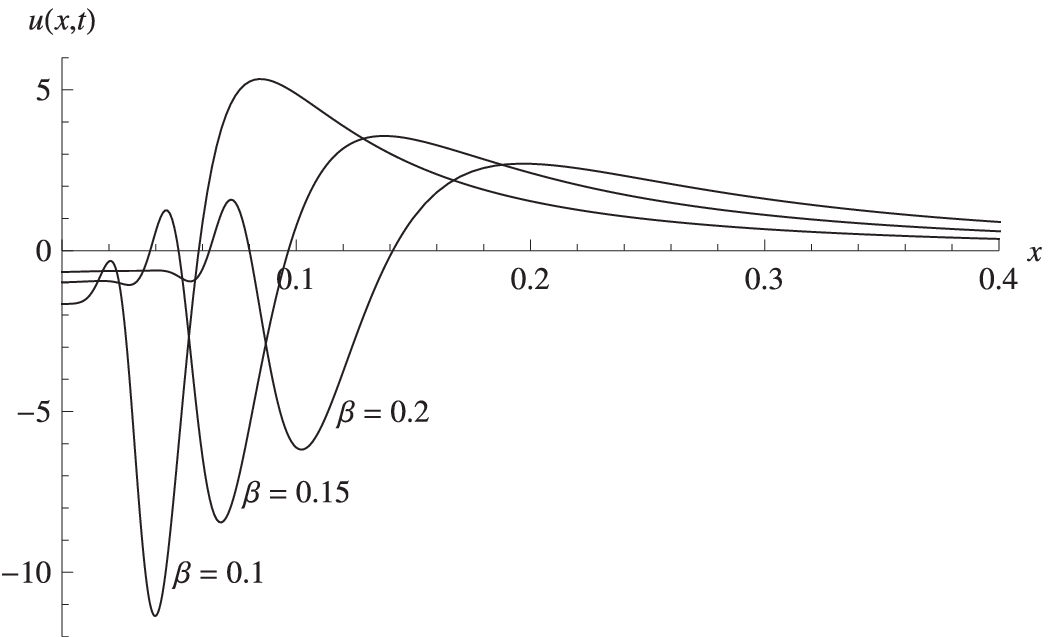}
 \caption{Displacement $u\left( x,t\right) $ at $t=1$ as a function of $x$ as a solution of (\ref{eq:fzwe}).}
 \label{fig.4}
 \end{minipage}
\hfil
\begin{minipage}{65mm}
 \centering
 \includegraphics[scale=0.7]{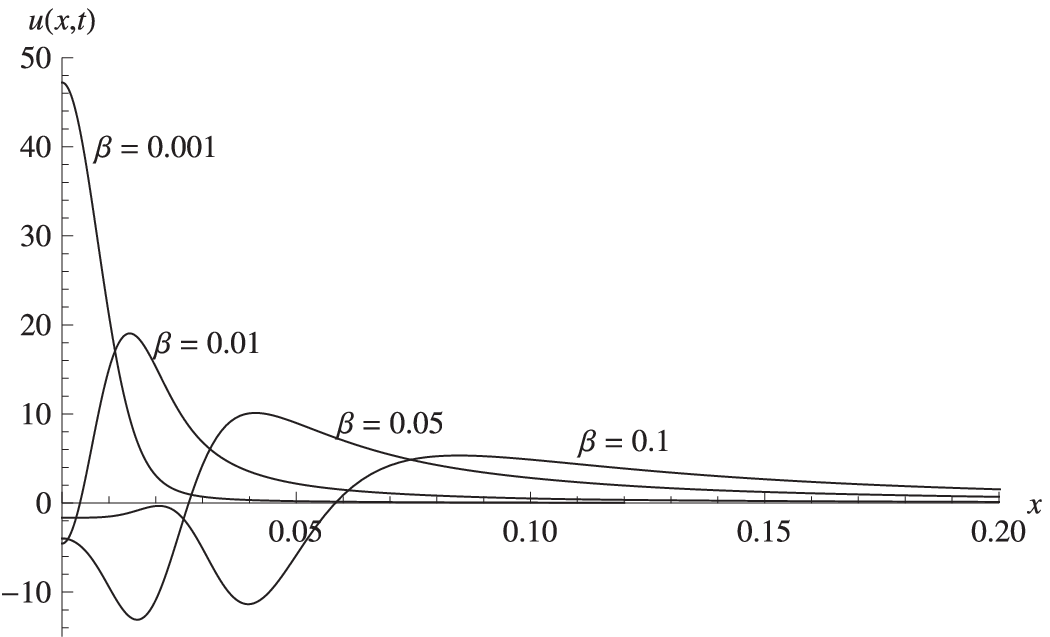}
 \caption{Displacement $u\left( x,t\right) $ at $t=1$ as a function of $x$ as a solution of (\ref{eq:fzwe}).}
 \label{fig.5}
 \end{minipage}
\end{figure}
In Figure \ref{fig.4} one notices the significant influence of non-local
effects. Namely, the secondary peaks are more prominent in height than the
primary peak. Finally, in the limiting case when $\beta \rightarrow 0$ one
expects to obtain the displacement field in the form (\ref{non-prop}), i.e.,
$u\left( x,t\right) =\delta \left( x\right) ,$ $x\in \mathbb{R},$ $t>0.$
This can be seen from Figure \ref{fig.5}. We might, thus, say that these
numerical examples supports the claim that $\beta ,$ as the non-locality
parameter, measures the resistance of material to the disturbance
propagation. Namely, at the same time instant, as $\beta $ decreases, the
primary peak is placed closer to the point where the Dirac-type disturbance
occurred and for $\beta \rightarrow 0$ we obtain the non-propagating
disturbance. Also, the shape of the primary peaks changes and, as $\beta $
increases, they become more alike the peak of the local model and for $\beta
=1$ we have the overlap of the curves.

\appendix

\section{Mathematical background \label{sec:math prelim}}

This section serves as a mathematical survey needed in analysis that we have
presented. We single out definitions and properties of fractional
derivatives and since our main tools are integral transforms, we recall,
more or less well-known, main definitions and properties used. For a
detailed exposition of the theory of fractional calculus see \cite{Pod,SKM},
and for the spaces and integral transforms we refer to \cite{SKM,vlad}.

Let $0\leq \alpha <1$, $-\infty \leq a<b\leq \infty $. The left and right
Caputo derivatives, of order $\alpha $, of an absolutely continuous function
$u$ are defined by
\begin{equation}
{}_{a}^{C}\mathrm{D}_{t}^{\alpha }u(t)=\frac{1}{\Gamma (1-\alpha )}%
\int_{a}^{t}\frac{u^{\prime }(\theta )}{(t-\theta )^{\alpha }}\mathrm{d}%
\theta ,\;\;\text{and}\;\;{}_{t}^{C}\mathrm{D}_{b}^{\alpha }u(t)=-\frac{1}{%
\Gamma (1-\alpha )}\int_{t}^{b}\frac{u^{\prime }(\theta )}{(\theta
-t)^{\alpha }}\mathrm{d}\theta ,  \label{eq:Caputo-levi}
\end{equation}%
where $\Gamma $ is the Euler gamma function and $u^{\prime }=\frac{\mathrm{d}%
}{\mathrm{d}t}u$. Note that ${}_{a}^{C}\mathrm{D}_{t}^{0}u(t)={}_{t}^{C}%
\mathrm{D}_{b}^{0}u(t)=u(t)$, and for continuously differentiable functions
and distributions we have that as $\alpha \rightarrow 1^{-}$, ${}_{a}^{C}%
\mathrm{D}_{t}^{1}u(t)\rightarrow u^{\prime }(t)$, ${}_{t}^{C}\mathrm{D}%
_{b}^{1}u(t)\rightarrow -u^{\prime }(t)$. Therefore, the Caputo derivatives
generalize integer order derivatives.

Let $0\leq \beta <1$, $-\infty \leq a<b\leq \infty $. The symmetrized
fractional derivative of an absolutely continuous function $u$ is defined as
\begin{eqnarray}
{}_{a}^{C}\mathcal{E}_{b}^{\beta }u(x) &=&\frac{1}{2}\left( {}_{a}^{C}%
\mathrm{D}_{x}^{\beta }-{}_{x}^{C}\mathrm{D}_{b}^{\beta }\right) u(x)
\label{eq:Caputo-simetrizovani} \\
&=&\frac{1}{2}\frac{1}{\Gamma (1-\beta )}\int_{a}^{b}\frac{u^{\prime
}(\theta )}{|x-\theta |^{\beta }}\mathrm{d}\theta .  \notag
\end{eqnarray}%
For $a=-\infty $ and $b=\infty $ we write $\mathcal{E}_{x}^{\beta }$ instead
of ${}_{a}^{C}\mathcal{E}_{b}^{\beta }$ and then
\begin{equation*}
\mathcal{E}_{x}^{\beta }u(x)=\frac{1}{2}\frac{1}{\Gamma (1-\beta )}%
|x|^{-\beta }\ast u^{\prime }(x).
\end{equation*}%
Note that $\mathcal{E}_{x}^{0}u(x)=0$ and $\mathcal{E}_{x}^{\beta
}u(x)\rightarrow u^{\prime }(x)$, as $\beta \rightarrow 1$. So, the
symmetrized fractional derivative generalizes the first derivative of a
function. The zeroth order symmetrized fractional derivative of a function
is zero (not a function itself).

For fractional operators in the distributional setting, one introduces a
family $\{f_{\alpha }\}_{\alpha \in \mathbb{R}}\in \mathcal{S}_{+}^{\prime }$
as
\begin{equation*}
f_{\alpha }(t)=\left\{
\begin{array}{ll}
H(t)\frac{t^{\alpha -1}}{\Gamma (\alpha )}, & \alpha >0, \\
\frac{\mathrm{d}^{N}}{\mathrm{d}t^{N}}f_{\alpha +N}(t), & \alpha \leq
0,\alpha +N>0,N\in \mathbb{N},%
\end{array}%
\right.
\end{equation*}%
and $\{\check{f}_{\alpha }\}_{\alpha \in \mathbb{R}}\in \mathcal{S}%
_{-}^{\prime }$ as
\begin{equation*}
\check{f}_{\alpha }(t)=f_{\alpha }(-t),
\end{equation*}%
where $H$ is the Heaviside function. Then $f_{\alpha }\ast $ and $\check{f}%
_{\alpha }\ast $ are convolution operators and for $\alpha <0$ they are
operators of left and right fractional differentiation, so that for $u$
absolutely continuous we have
\begin{equation*}
{}_{0}^{C}\mathrm{D}_{t}^{\alpha }u=f_{1-\alpha }\ast u^{\prime }\;\;\text{%
and}\;\;{}_{t}^{C}\mathrm{D}_{a}^{\alpha }u=-\check{f}_{1-\alpha }\ast
u^{\prime }.
\end{equation*}

For $u\in \mathcal{S}^{\prime }$ the Fourier transform is defined as
\begin{equation*}
\left\langle \hat{u},\varphi \right\rangle =\left\langle u,\hat{\varphi}%
\right\rangle ,\;\;\varphi \in \mathcal{S}(\mathbb{R}),
\end{equation*}%
where for $\varphi \in \mathcal{S}$
\begin{equation*}
\hat{\varphi}(\xi )=\mathcal{F}\left[ \varphi \left( x\right) \right] (\xi
)=\int_{-\infty }^{\infty }\varphi (x)\mathrm{e}^{-\mathrm{i}\xi x}\mathrm{d}%
x,\;\;\xi \in \mathbb{R}.
\end{equation*}

The Laplace transform of $u\in \mathcal{S}^{\prime }$ is defined by
\begin{equation*}
\tilde{u}(s)=\mathcal{L}\left[ u\left( t\right) \right] (s)=\mathcal{F}\left[
\mathrm{e}^{-\xi t}u\left( t\right) \right] (\eta ),\;\;s=\xi +\mathrm{i}%
\eta .
\end{equation*}%
It is well known that the function $\tilde{u}$ is holomorphic in the half
plane $\func{Re}s>0$, see e.g. \cite{vlad}. In particular, for $u\in L^{1}(%
\mathbb{R})$ such that $u(t)=0$, for $t<0$, and $|u(t)|\leq A\mathrm{e}^{at}$
($a,A>0$) the Laplace transform is
\begin{equation*}
\tilde{u}(s)=\int_{0}^{\infty }u(t)\mathrm{e}^{-st}\mathrm{d}t,\;\;\func{Re}%
s>0.
\end{equation*}

We recall main properties of the Fourier and Laplace transforms. Let $%
u,u_{1},u_{2}\in \mathcal{S}^{\prime }$%
\begin{align*}
\mathcal{F}& [u_{1}\ast u_{2}](\xi )=\mathcal{F}u_{1}(\xi )\cdot \mathcal{F}%
u_{1}(\xi ),\quad \mathcal{F}\left[ u^{(n)}\right] (\xi )=(\mathrm{i}\xi
)^{n}\mathcal{F}y(\xi ),\,n\in \mathbb{N},\quad \mathcal{F}\delta (s)=1, \\
\mathcal{L}& [u_{1}\ast u_{2}](s)=\mathcal{L}u_{1}(s)\cdot \mathcal{L}%
u_{1}(s),\quad \mathcal{L}[{}_{0}\mathrm{D}_{t}^{\alpha }u](s)=s^{\alpha }%
\mathcal{L}u(s),\,\alpha \geq 0,\quad \mathcal{L}\delta (s)=1,
\end{align*}%
where $(\cdot )^{(n)}$ denotes $n$-th derivative. For $\beta \in \left[
0,1\right) $ it holds%
\begin{gather*}
\mathcal{F}\left[ |x|^{-\beta }\right] (\xi )=2\Gamma (1-\beta )\sin \frac{%
\beta \pi }{2}\frac{1}{|\xi |^{1-\beta }}, \\
\mathcal{F}\left[ \mathcal{E}_{x}^{\beta }u\left( x\right) \right] (\xi )=%
\mathrm{i}\frac{\xi }{|\xi |^{1-\beta }}\sin \frac{\beta \pi }{2}\hat{u}(\xi
).
\end{gather*}

\section*{Acknowledgement}

This research is supported by the Serbian Ministry of Education and Science
projects $174005$, $174024,$ III$44003$ and TR$32035$, as well as by the
Secretariat for Science of Vojvodina project $114-451-3605/2013$.

The authors would also like to thank to Radovan Obradovi\'{c} for his
valuable suggestions concerning the calculations in numerical examples.


\end{document}